\documentclass[a4paper]{article}

\usepackage{amsmath, amssymb}
\usepackage{graphicx, subfigure}
\usepackage{color}
\usepackage{booktabs}
\usepackage{rotating}
\usepackage[authoryear,round]{natbib}

\newcommand{\keywords}[1]{{\bf Keywords:} {#1}}

\newcommand{\JEL}[1]{{\bf JEL:} {#1}}

\newtheorem{thm}{theorem}[section]

\newtheorem{proposition}[thm]{Proposition}

\newtheorem{lemma}[thm]{Lemma}
\newenvironment{proof}{{\bf Proof. }}{\hfill$\Box$}
\newenvironment{MyProof}[1]{{\bf Proof {#1}}}{\hfill$\Box$}

\newcommand{\MakeTitle}{\maketitle\newcommand{\and}{$\cdot$ }}

\title{Seasonal Stochastic Volatility and Correlation together with the Samuelson Effect in Commodity Futures Markets
	\thanks{
		We would like to thank
        Ennio Fedrizzi, Fran\c{c}ois Le Grand and Cassio Neri
        for helpful and stimulating comments, discussions and suggestions.
	}
}

\author{
	Lorenz Schneider
	\thanks{Center for Financial Risks Analysis (CEFRA), EMLYON Business School, \texttt{schneider@em-lyon.com}.}
    \quad \quad
	Bertrand Tavin
	\thanks{Center for Financial Risks Analysis (CEFRA), EMLYON Business School, \texttt{tavin@em-lyon.com}.}
}

\date{\today}

\begin{document}

\MakeTitle


\begin{abstract} %
We introduce a multi-factor stochastic volatility model based on the CIR/Heston volatility process that incorporates seasonality and the Samuelson effect.
First, we give conditions on the seasonal term under which the corresponding volatility factor is well-defined.
These conditions appear to be rather mild.
Second, we calculate the joint characteristic function of two futures prices for different maturities in the proposed model.
This characteristic function is analytic.
Finally, we provide numerical illustrations in terms of implied volatility and correlation produced by the proposed model with five different specifications of the seasonality pattern.
The model is found to be able to produce volatility smiles at the same time as a volatility term-structure that exhibits the Samuelson effect with a seasonal component.
Correlation, instantaneous or implied from calendar spread option prices via a Gaussian copula, is also found to be seasonal.
\bigskip


\bigskip





\JEL{C63 \and C52 \and G13}
\end{abstract}


\section{Introduction}
\label{s:Introduction}

Seasonality is a well-known empirical feature of several commodities markets.
In the energy sector, among fossil fuels, natural gas futures curves,
and among refined products, gasoline, heating oil and fuel oil futures curves all typically display seasonality.
In the agricultural sector, almost all futures curves show seasonality due to harvest times and the seasons of the year.

It is important to distinguish from the outset between two types of seasonality:
seasonality of futures prices and seasonality of volatility of futures prices.

Regarding seasonality of prices, consider agricultural commodities such as corn, soybeans and wheat.
These tend to be in high supply after the harvest in summer, and in low supply in the months preceding the harvest.
This typically leads to relatively low prices of futures contracts with delivery months in the summer or early fall,
and high futures prices of contracts with delivery months in late winter or spring.
Therefore, when the prices of these contracts are plotted as a function of their maturity, they tend to rise and fall with the maturity in some seasonal way.
In other words, the futures curve shows seasonality.
However, the price of an individual futures contract with a given maturity should not rise and fall over time in any kind of seasonal way:
indeed, such a behaviour would lead to easy arbitrage opportunities.

Regarding seasonality of volatility, the situation is different in the sense that now an individual futures contract, with fixed maturity,
tends to go through phases of relatively high or low volatility according to a seasonal pattern.
To take again the example of agricultural commodities, the weather in the months leading up to the harvest has a direct impact on
its quality and quantity, and futures prices can fluctuate strongly as forecasts for the new crop change.
In contrast to this, weather patterns in winter tend to be of minor consequence for the harvest, and futures prices tend to fluctuate less strongly.

It follows from these empirical observations that for commodity models, seasonality is usually only an issue for the volatility, but not for the futures price itself.
Mathematically, individual futures prices are modelled as martingales, and martingales do not have a tendency to rise or fall in a pre-determined way.
\citet{Clark2014} gives a general discussion and numerous examples of seasonality in various commodities markets.

Traditionally, there are two approaches to modelling the prices of futures contacts: futures-based models and spot-based models.
An advantage of futures-based models models is that since the futures price curve is an input of the model,
any arbitrage-free shape of the initial futures curve can be accommodated, including any type of seasonality.
In contrast, a first step for spot-based models is to make them fit the initial futures curve, which uses up model parameters and doesn't necessarily
lead to satisfactory results.

\citet{Sorensen2002} studies the modelling of seasonality in corn, soybean and wheat futures markets.
Analysis of a large data set of CBOT futures prices data from 1972 to 1997 confirms clearly that futures prices exhibit a seasonality.
Another feature that is suggested by the data is seasonal behaviour of the futures price volatilities.
In this vein, \citet{RichterSorensen2002} propose a model for the spot price of soybeans based on seasonal stochastic volatility.
\citet{GemanNguyen2005} also introduce a spot-based model for soybean prices with seasonality both for the price level and the (possibly stochastic) volatility level.
\citet{BackProkopczukRudolf2013} analyze data from corn, soybean, heating oil and natural gas markets and compare various spot-based models with deterministic seasonal volatility.
They conclude that a volatility with seasonality is an important feature when valuing options on futures in these markets.
\citet{BackProkopczukRudolf2011} also study a futures-based model with seasonal stochastic volatility, which is essentially the \citet{Heston1993} stochastic volatility model
with deterministic, seasonal mean-reversion level in the square-root process followed by the variance.
\citet{SchmitzWangKimn2013} study calendar spread options in agricultural grain markets relying on a joint Heston model for the two underlying futures contracts.
These two contracts share the same variance process, which has a constant mean-reversion level, and therefore does not display seasonality.
In the context of interest rates, the \citet{CoxIngersollRoss1985} (CIR) model has been extended to time-dependent parameters by \citet{Maghsoodi1996},
and the \citet{Heston1993} model with time-dependent parameters, including the correlation between the spot price and its variance, has been studied by \citet{BenhamouGobetMiri2010}.
Let us also note that in the context of electricity markets, \citet{LuciaSchwartz2002} give a detailed justification of the choice of seasonality function,
as do \citet{GemanRoncoroni2006}.

In parallel to his remark about seasonality in futures prices, \citet{Sorensen2002} confirms the \citet{Samuelson1965} hypothesis that
``the variations of distant maturity futures are lower than nearby futures prices.''
We call this pattern the \textit{Samuelson effect.} 
Popular futures-based models that incorporate this effect are those of \citet{ClewlowStrickland1999_1,ClewlowStrickland1999_2}.
The volatility functions used in these models are deterministic.
\citet{SchneiderTavin2015} extend the multi-factor model of \citet{ClewlowStrickland1999_2} to incorporate stochastic volatility.
Not only is stochastic volatility an incontestible empirical feature of prices in futures markets,
its inclusion also allows to calibrate the model to option volatility smiles and skews typically seen in futures option markets.
In agricultural markets, a reflection of stochastic volatility is the recent introduction of several volatility indices on the CBOE/CBOT:
the Corn Volatility Index (CIV) and Soybean Volatility Index (SIV) were introduced in 2011,
and the Wheat Volatility Index (WIV) was introduced in 2012.

In this paper, we extend the model introduced in \citet{SchneiderTavin2015} to incorporate seasonal trends in the stochastic volatility processes.
To achieve this, we begin by studying the mathematical conditions to impose on the seasonality function to guarantee that the generalized CIR process
retains important features, such as existence and uniqueness of a strong solution, and positivity.
It turns out that these conditions are very mild.
These conditions appear to be not only interesting from a theoretical point of view,
but also useful in practice, because different markets may need to be modelled with different seasonality patterns for the volatility.

We then introduce the model with seasonal volatility and show how, by a generalisation of the results in \citet{SchneiderTavin2015},
the joint characteristic function of the log-returns of two futures prices can be obtained.
It turns out that the Riccati ODE for the first function $A$ is not affected, and only the integral ODE for the second function $B$ depends on $\theta$ and is altered.
Therefore, the same closed-form solution for $A$ as in \citet{SchneiderTavin2015} can be used.

Next, we propose several specifications of seasonality functions and compare them.

Then, we calculate implied volatility surfaces in our model.
We also calculate calendar spread option prices and examine the effect of changes in the seasonality function on these prices.

Regarding correlations, we study the effect seasonality has on the instantaneous correlation between two futures contracts in the case when the variances are deterministic.
We find that the influence of the seasonality on the correlation depends on the magnitude of the exponential damping of the volatility factors.
And from the calendar spread option prices we calculate implied correlations, where again we can observe how the seasonal pattern of the variance
translates into a seasonal pattern of the implied correlation.


The rest of the paper proceeds as follows.
In Section \ref{s:TheCIRProcessWithTimeDependentDrift} we discuss the CIR process with time-dependent drift.
In Section \ref{s:SeasonalStochasticVolatilityModelForAgriculturals} we define the proposed model and calculate the associated joint characteristic function.
We also give the methods to be used for vanilla and calendar spread option pricing.
In Section \ref{s:SeasonalityFunctions} we review various seasonality functions that can be used to specify the proposed model.
Section \ref{s:ImpliedVolatilities} gives a numerical illustration of the implied volatility patterns produced by the proposed model with different seasonality functions.
Section \ref{s:SeasonalStochasticCorrelation} deals with the seasonal behaviour of the instantaneous correlation when volatility is seasonal.
In this section we also study the effect of seasonality on calendar spread option prices.
Section \ref{s:Conclusion} concludes.


\section{The CIR Process with Time-Dependent Drift}
\label{s:TheCIRProcessWithTimeDependentDrift}

To our knowledge, \citet{HullWhite1990} were the first to consider extending the \citet{CoxIngersollRoss1985} (CIR) interest rate model to time-dependent coefficients.
They conclude that in this general case, it is no longer possible to obtain European bond option prices analytically.
\citet{Maghsoodi1996} also studies the ``extended'' CIR process in which the parameters $\kappa, \theta$ and $\sigma$ are time-dependent and finds, under certain conditions,
the unique strong solution to the SDE describing the evolution of the process.

In the context of the \citet{Heston1993} stochastic volatility model, the CIR process represents the variance process of a stock price or foreign-exchange rate.
\citet{BenhamouGobetMiri2010} study the ``time dependent Heston model'' and derive analytical formulas approximating European option prices.
In their setup, the mean-reversion parameter $\kappa$ is constant, but the parameters $\theta, \sigma$ and $\rho$ (giving the correlation between the stock price, or foreign-exchange rate, and
its variance) are all allowed to vary with time $t$.

In the model introduced here, we only let the mean-reversion level given by $\theta$ depend on time,
while the other parameters $\kappa > 0$ and $\sigma > 0$ (and later also $\rho$) remain constant.

Let $(\Omega, \mathcal{A}, {\mathbb P}, \mathcal{F})$ be a filtered probability space, and let $B = (B_t)_{t \geq 0}$ be a Brownian motion on this space.
Let $\mathcal{T} = \{ t_i, i = 1, ... \}$ be a set of times having only finitely many points in every bounded interval,
and let $\mathcal{Z} = \{ 0 \leq t_1 < t_2 < ... < t_i < ... \}$ be the partition of $\mathbb{R}_0^+$ defined by $\mathcal{T}$.
Finally, let the seasonality function $\theta: \mathbb{R}_0^+ \to \mathbb{R}^+$ be piecewise continuous with respect to $\mathcal{Z}$,
and assume that it is bounded from below and above by positive constants $\theta_{min}$ and $\theta_{max}$.

We will compare two processes $v$ (seasonal) and $\tilde{v}$ (non-seasonal), which are given, respectively, by the SDEs
\begin{align}
\label{CIR_SDE_TimeDependentTheta}
dv(t) &= \kappa \left( \theta(t) - v(t) \right) dt + \sigma \sqrt{v(t)} dB(t),
\\
\label{CIR_SDE_ConstantTheta}
d \tilde{v}(t) &= \kappa \left( \theta_{min} - \tilde{v}(t) \right) dt + \sigma \sqrt{\tilde{v}(t)} dB(t),
\end{align}
with identical parameters $\kappa > 0, \sigma > 0$ and initial conditions $0 < \tilde{v}(0) = \tilde{v}_0 \leq v(0) = v_0$.

It is well known that \eqref{CIR_SDE_ConstantTheta} has a unique strong solution.
The following result describes the solution to \eqref{CIR_SDE_TimeDependentTheta}.
\begin{proposition}
\label{Prop:CIR_SDE_Solution}
Assume that the seasonality function $\theta$ is piecewise continuous w.r.t. the partition $\mathcal{Z}$ of $\mathbb{R}_0^+$,
and bounded by positive constants $\theta_{min}$ and $\theta_{max}$, i.e. for all $t \geq 0, 0 < \theta_{min} \leq \theta(t) \leq \theta_{max}$.
Let the processes $v$ and $\tilde{v}$ be given by \eqref{CIR_SDE_TimeDependentTheta} and \eqref{CIR_SDE_ConstantTheta}, respectively.
Then:
\begin{enumerate}
\item
The process \eqref{CIR_SDE_TimeDependentTheta} has a unique strong solution with continuous sample paths.
\item
${\mathbb P} \left[ \tilde{v}_t \leq v_t, \forall t \geq 0 \right] = 1$.
\item
If the Feller condition $\sigma^2 < 2 \kappa \theta_{min}$ is satisfied for $\theta_{min}$, then
the process $v$ is strictly positive.
\end{enumerate}
\end{proposition}

We prove this result in appendix \ref{a1:Proofs}.

Note that if the Feller condition is violated, then $\tilde{v}$ can possibly reach $0$, but it still cannot become negative.

The piecewise continuity condition on $\theta$ means that even discontinuous specifications of the mean-reversion level,
such as the sawtooth function given by
\begin{equation}
\label{SawtoothFunction}
\theta(t) = a + b \left( t - t_0 - \left \lfloor t - t_0 \right \rfloor \right),
\end{equation}
with $a,b > 0$ and $t_0 \in [0,1[$, pose no problems.
Some other examples of the form of $\theta$ are discussed in Section \ref{s:SeasonalityFunctions}.


\section{A Model with Seasonal Stochastic Volatility for Agricultural Futures}
\label{s:SeasonalStochasticVolatilityModelForAgriculturals}

\subsection{The Financial Framework and the Model}
\label{ss:FinancialFrameworkAndModel}

We begin by giving a mathematical description of our model under the risk-neutral measure ${\mathbb Q}$.
Let $n \geq 1$ be an integer, and let $B_1, ..., B_{2n}$ be Brownian motions under ${\mathbb Q}$.
Let $T_m$ be the maturity of a given futures contract.
The futures price $F(t, T_m)$ at time $t, 0 \leq t \leq T_m$, is assumed to follow the stochastic differential equation (SDE)
\begin{equation}
\label{FuturesSDE}
dF(t, T_m) = F(t, T_m) \sum_{j=1}^n e^{-\lambda_j(T_m - t)} \sqrt{v_j(t)} dB_j(t), \; F(0, T_m) = F_{m,0} > 0.
\end{equation}
The processes $v_j, j=1, ..., n,$ are stochastic variance processes with time-dependent seasonal mean-reversion level assumed to follow the SDE
\begin{equation}
\label{VarianceSDE}
dv_j(t) = \kappa_j \left( \theta_j(t) - v_j(t) \right) dt + \sigma_j \sqrt{v_j(t)} dB_{n+j}(t), \; v_j(0) = v_{j,0} > 0.
\end{equation}
Various possibilities of the specification of the seasonal mean-reversion level functions $\theta_j: \mathbb{R}_0^+ \to \mathbb{R}^+$ are presented and discussed in Section \ref{s:SeasonalityFunctions}.
Note that the initial futures curve $F(0, T_m), m = 1, 2, ...,$ is exogenous in our model and can therefore accommodate any seasonal pattern shown by the futures prices.

For the correlations, we assume
\begin{equation}
\label{FuturesVarianceCorrelations}
\langle dB_j(t), dB_{n+j}(t) \rangle = \rho_j dt, -1 < \rho_j < 1, j=1, ..., n,
\end{equation}
and that otherwise the Brownian motions $B_j, B_k, k \neq j, j + n,$ are independent of each other.
As we will see, this assumption has as a consequence that the characteristic function factors into $n$ separate expectations.

For fixed $T_m$, the futures log-price $\ln F(t, T_m)$ follows the SDE
\begin{equation}
\label{LogFuturesSDE}
d\ln F(t, T_m) = \sum_{j=1}^n \left( e^{-\lambda_j(T_m - t)} \sqrt{v_j(t)} dB_j(t) - \frac{1}{2} e^{-2\lambda_j(T_m - t)} v_j(t) dt \right),
\; \ln F(0, T_m) = \ln F_{m,0}.
\end{equation}
Integrating \eqref{LogFuturesSDE} from time $0$ up to a time $T, T \leq T_m$, gives
\begin{equation}
\label{LogFuturesIntegratedSDE}
\ln F(T, T_m) -  \ln F(0, T_m) = \sum_{j=1}^n \int_0^T e^{-\lambda_j(T_m - t)} \sqrt{v_j(t)} dB_j(t) - \frac{1}{2} \sum_{j=1}^n \int_0^T e^{-2\lambda_j(T_m - t)} v_j(t) dt.
\end{equation}

We define the log-return between times $0$ and $T$ of a futures contract with maturity $T_m$ as
$$
X_m(T) := \ln \left( \frac{F(T, T_m)}{F(0, T_m)} \right).
$$
In the following, the joint characteristic function $\phi$ of two log-returns $X_1(T), X_2(T)$ will play an important role.
For $u = (u_1, u_2) \in {\mathbb C}^2$, $\phi$ is given by
\begin{equation}
\label{jointCharacteristicFunction_returns}
\phi(u)
= \phi(u; T, T_1, T_2)
= {\mathbb{E^Q}} \left[ \exp \left( i \sum_{k=1}^2 u_k X_k(T) \right) \right].
\end{equation}
The joint characteristic function $\Phi$ of the futures log-prices $\ln F(T, T_1), \ln F(T, T_2)$ is then given by
\begin{equation}
\label{jointCharacteristicFunction_prices}
\Phi(u) = \exp \left( i \sum_{k=1}^2 u_k \ln F(0, T_k) \right) \cdot \phi(u).
\end{equation}
Note that futures prices in our model are not mean-reverting,
and that the log-price $\ln F(t, T_m)$ at time $t$ and the log-return $\ln F(T, T_m) - \ln F(t, T_m)$ are independent random variables.

In the following proposition, we show how the joint characteristic function $\phi$, and therefore also the single characteristic function $\phi_1$,
is given by a system of two ordinary differential equations (ODE).
\begin{proposition}
\label{Prop:JointCharacteristicFunction}
The joint characteristic function $\phi$ at time $T \leq T_1, T_2$ for the log-returns $X_1(T), X_2(T)$ of two futures contracts with maturities $T_1, T_2$ is given by
\begin{align*}
\phi(u) &= \phi(u; T, T_1, T_2)
\\
&=
\prod_{j=1}^n
\exp \left( -i \frac{\rho_j}{\sigma_j} f_{j,1}(u,0) \left( v_j(0) + \kappa_j \hat{\theta}_{j, T} \right) \right)
\exp \left( A_j(0,T) v_j(0) + B_j(0,T) \right),
\end{align*}
where
\begin{align*}
f_{j,1}(u,t)        &= \sum_{k=1}^2 u_k e^{- \lambda_j(T_k - t)}, \quad f_{j,2}(u,t) = \sum_{k=1}^2 u_k e^{-2\lambda_j(T_k - t)},
\\
q_j(u,t)            &= i \rho_j \frac{\kappa_j - \lambda_j}{\sigma_j} f_{j,1}(u,t) - \frac{1}{2} (1 - \rho_j^2) f_{j,1}^2(u,t) - \frac{1}{2} i f_{j,2}(u,t),
\\
\hat{\theta}_{j, T} &= \int_0^T \theta_j(t) e^{\lambda_j t} dt,
\end{align*}
and the functions $A_j: (t,T) \mapsto A_j(t,T)$ and $B_j: (t,T) \mapsto B_j(t,T)$ satisfy the two differential equations
\begin{align*}
\frac{\partial A_j}{\partial t} - \kappa_j A_j + \frac{1}{2} \sigma_j^2 A_j^2 + q_j &= 0,
\\
\frac{\partial B_j}{\partial t} + \kappa_j \theta_j(t) A_j &= 0,
\end{align*}
with $A_j(T,T) = i \frac{\rho_j}{\sigma_j} f_{j,1}(u,T), \; B_j(T,T) = 0.$

The single characteristic function $\phi_1$ at time $T \leq T_1$ for the log-return $X_1(T)$ of a futures contract with maturity $T_1$ is given by setting $u_2 = 0$
in the joint characteristic function.
\end{proposition}

Note that the integrals $\hat{\theta}_{j, T}$ only depend on the specification of the seasonality functions $\theta_j$ and the maturity $T$.
Therefore, their value can be calculated once and then stored, avoiding recalculations during repeated calls to the characteristic function.
Of course, if $\theta$ is a constant function, then the joint characteristic function given above is the same as the one given in \citet{SchneiderTavin2015}.

We prove this result in appendix \ref{a1:Proofs}.

\subsection{Pricing Vanilla Options}
\label{ss:PricingVanillaOptions}


European options on futures contracts can be priced using the Fourier inversion technique as described in \citet{Heston1993} and \citet{BakshiMadan2000}.
Let $K$ denote the strike and $T$ the maturity of a European call option on a futures contract $F$ with maturity $T_m \geq T$.
The function needed for this technique is the single characteristic function $\Phi_1$ of the futures log-price $\ln F(T, T_m)$, given by
$\Phi_1(u) = e^{i u \ln F(0, T_m)} \phi_1(u)$,
with $\phi_1(u)$ obtained from Proposition \ref{Prop:JointCharacteristicFunction}.
European put options can be priced via put-call parity $C - P = e^{-rT} \left( F(0, T_1) - K \right)$.

\subsection{Pricing Calendar Spread Options}
\label{ss:PricingCalendarSpreadOptions}



Calendar spread options (CSO) are popular options in commodity markets.
To give a recent example, in February 2015 the Minneapolis Grain Exchange (MGX) introduced North American Hard Red Spring Wheat (HRSW) CSOs for trade on CME Globex.
Calendar spread options are defined as follows.

Let two futures maturities $T_1, T_2$, an option maturity $T$, and a strike $K$ (which is allowed to be negative) be fixed.
Then the payoffs of calendar spread call and put options, $CSC$ and $CSP$, are respectively given by
\begin{equation}
\label{CalendarSpreadCallPayoff}
CSC(T) = \left( F(T, T_1) - F(T, T_2) - K \right)^+,
\end{equation}
\begin{equation}
\label{CalendarSpreadPutPayoff}
CSP(T) = \left( K - \left( F(T, T_1) - F(T, T_2) \right) \right)^+ .
\end{equation}
To evaluate such options with a pricing model, the discounted expectation of the payoff must be calculated in the risk-neutral measure.
As for ``vanilla'' European options, there is a model-independent put-call parity for calendar spread options:
\begin{equation}
\label{CalendarSpreadPutCallParity}
CSC(0) - CSP(0) = e^{-rT} \left( F(0, T_1) - F(0, T_2) - K \right).
\end{equation}
\citet{CaldanaFusai2013} show how to calculate CSC prices with models for which the joint characteristic function is known.
Strictly speaking, their methods give a lower bound for the calendar spread option price, but usually this bound is very close to the true price.
Note that in case the strike $K = 0$ the formula is exact.
When we calculate CSO prices, we do so with this method.
An alternative method, based on the $2$-dimensional FFT algorithm, is given by \citet{HurdZhou2010}.
We refer to \citet{SchneiderTavin2015} and references therein for more details on calendar spread options.


\section{Seasonality Functions}
\label{s:SeasonalityFunctions}

Below we present four types of seasonality functions that can be used as parametric forms to model seasonal variations of the volatility.
For a given factor, $\theta$ represents the seasonality term of the volatility dynamics and $\hat{\theta}$ is an integral involving $\theta$ appearing in different expressions such as the characteristic function.
For $T>0$ and $\lambda \in \mathbb{R}$, $\hat{\theta}$ is written
\begin{equation}
\label{thetaTransform}
\hat{\theta}_T(\lambda) = \int^{T}_{0}{\theta(t)e^{\lambda t}dt}.
\end{equation}

The presented seasonality functions $\theta$ are parametric and work with three parameters: $a,b$ and $t_0$.
Parameter $a$ controls for the volatility level.
Parameter $b$ controls the magnitude of the seasonality pattern and $t_0$ corresponds to the time of the year when the volatility reaches its maximum.

It seems relevant to consider different seasonality patterns as the reasons underpinning the seasonality phenomena in volatility may vary from a market to another.
The first two patterns considered below are smooth and are based on the sinus function.
The three others have points of non-differentiability and/or discontinuity, and may be used to represent a less regular evolution of the volatility.

\underline{The sinusoidal pattern} is written, with $a,b > 0$ and $t_0 \in [0,1[$,
\begin{align}
\theta(t) & = a + b\cos{\left(2\pi\left(t - t_0\right)\right)},
\label{sinusoid} \\
\hat{\theta}_T(\lambda) & = \frac{b e^{\lambda T}}{\lambda^2+4\pi^2}\left(2 \pi \sin {\left(2 \pi (T-t_0)\right)} + \lambda \cos{\left(2 \pi (T-t_0)\right)} \right) \nonumber \\
& + \frac{b}{\lambda^2+4\pi^2}\left(2 \pi \sin{\left(2\pi t_0 \right) - \lambda \cos {\left(2\pi t_0 \right)}}  \right) + \frac{a}{\lambda}\left(e^{\lambda T}-1\right).
\label{int_sinusoid}
\end{align}
For a proof of this expression, we refer to Appendix \ref{a1:Proofs}.

\underline{The exponential-sinusoidal pattern} is written, with $a,b > 0$ and $t_0 \in [0,1[$,
\begin{equation}
\theta(t) = a\exp{\left(b \cos{\left(2\pi\left(t - t_0 \right)\right)}\right)}.
\end{equation}
This parametric form for $\theta$ is used in \citet{BackProkopczukRudolf2011}.
There is no closed form expression for $\hat{\theta}_T(\lambda)$.

\underline{The sawtooth pattern} is written, with $a,b > 0$ and $t_0 \in [0,1[$,
\begin{align}
\theta(t) & = a + b\left(t-t_0-\left\lfloor t-t_0 \right\rfloor \right), \\
\hat{\theta}_T(\lambda) & = \frac{1}{\lambda}\left(a+b\left(\frac{1}{\lambda}-t_0 \right) \right)-\frac{e^{-\lambda T}}{\lambda}\left(a+b\left(T+\frac{1}{\lambda}-t_0 \right) \right) \nonumber \\
 & -\frac{be^{\lambda t_0}}{\lambda} \left(\left\lfloor T-t_0 \right\rfloor e^{\lambda(T-t_0)} - \left(\sum^{\left\lfloor T-t_0 \right\rfloor}_{k=1}{e^{\lambda k}}\right)\mathbb{I}_{\left\{ T \geq t_0\right\}} + \mathbb{I}_{\left\{ T<t_0\right\}}+e^{-\lambda t_0}-1 \right),
\label{int_sawtooth}
\end{align}
where $\left\lfloor . \right\rfloor$ denotes the floor function, $\mathbb{I}$ is the indicator function and, by convention, $\sum^{p}_{k=1}{e^{\lambda k}}=0$ if $p<1$. The proof leading to this expression is found in Appendix \ref{a1:Proofs}.

\underline{The triangle pattern} is written, with $a,b > 0$ and $t_0 \in [0,1[$,
\begin{align}
& \theta(t) = a + b\left|\frac{1}{2}-\left(t-t_0-\left\lfloor t-t_0 \right\rfloor \right) \right| , \\
& \hat{\theta}_T(\lambda) = \frac{a}{\lambda}\left(e^{\lambda T}-1 \right) + \frac{b e^{\lambda t_0}}{\lambda}\left[ \left(z_2 + \left(\frac{2}{\lambda}e^{-\frac{\lambda}{2}}+e^{-\lambda t_0}(z_2-t_0) \right)\mathbb{I}_{\left\{ t_0 > \frac{1}{2} \right\}} - e^{-\lambda t_0}(z_2-t_0)\mathbb{I}_{\left\{ t_0 \leq \frac{1}{2} \right\}} \right)\right. \nonumber \\
& + \left(\left(\frac{2}{\lambda} e^{\frac{\lambda}{2}} + z_2 e^{\lambda}-z_1 \right)\sum^{n-1}_{k=0}{e^{\lambda k}}+ e^{\lambda n}\left(\left(\frac{2}{\lambda}e^{\frac{\lambda}{2}}-z_3 e^{\lambda \alpha} \right)\mathbb{I}_{\left\{ \alpha > \frac{1}{2} \right\}} + z_3 e^{\lambda \alpha}\mathbb{I}_{\left\{ \alpha \leq \frac{1}{2} \right\}} -z_1 \right) \right) \mathbb{I}_{\left\{ T \geq t_0\right\}} \nonumber \\
& + \left(e^{\lambda(T-t_0)}\left(z_2+T-t_0 \right)-z_2 \right)\mathbb{I}_{\left\{T-t_0 \in [-\frac{1}{2},0[ \right\}} \nonumber \\
& - \left. \left(\frac{2}{\lambda}e^{-\frac{\lambda}{2}} + e^{\lambda(T-t_0)}\left(z_2+T-t_0 \right) + z_2 \right)\mathbb{I}_{\left\{T-t_0 \in [-1,-\frac{1}{2}[ \right\}} \right],
\label{int_triangle}
\end{align}
with $n = \left\lfloor T-t_0 \right\rfloor$, $\alpha = T- t_0 - \left\lfloor T-t_0 \right\rfloor$, $z_1 = \frac{1}{2} + \frac{1}{\lambda}$, $z_2 = \frac{1}{2} - \frac{1}{\lambda}$ and $z_3 = z_1 - \alpha$ and with convention, $\sum^{p}_{k=0}{e^{\lambda k}}=0$ if $p<0$. The proof leading to this expression is found in Appendix \ref{a1:Proofs}.

\underline{The spiked pattern} is written, with $a,b > 0$ and $t_0 \in [0,1[$,
\begin{equation}
\theta(t) = a + b\left(\frac{2}{1+\left|\sin(\pi(t-t_0)) \right|}-1 \right)^2.
\end{equation}
This parametric form for $\theta$ can be found in \citet{GemanRoncoroni2006} where it is used to model the time varying intensity of a jump process.
There is no closed form expression for $\hat{\theta}_T(\lambda)$.
\medskip

Figure \ref{Fig:theta_plots} presents the plots of these seasonal patterns with $t_0=\frac{7}{12}$.

\begin{figure}[H]
\centering
	\includegraphics[height=5.0cm]{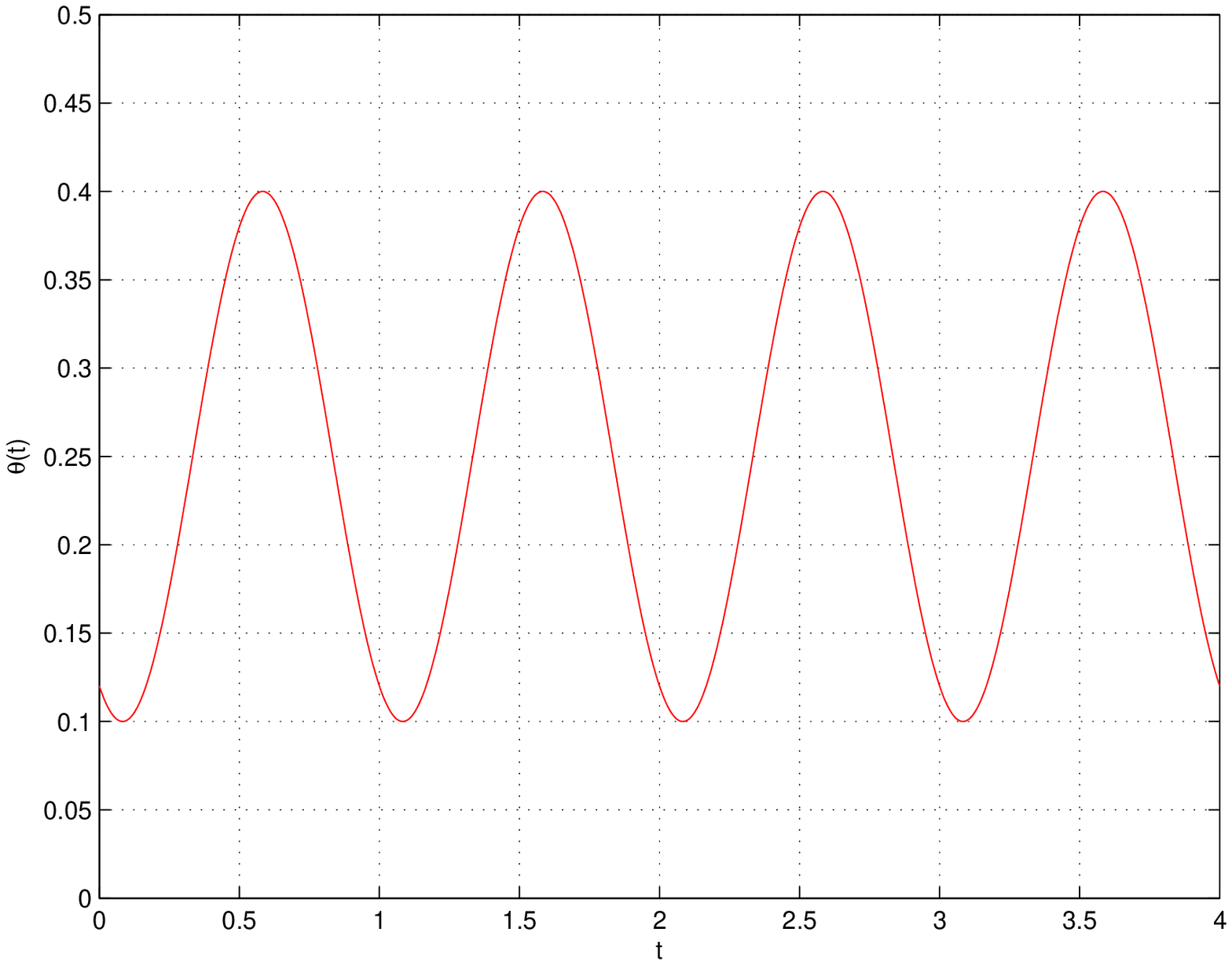} \quad
	\includegraphics[height=5.0cm]{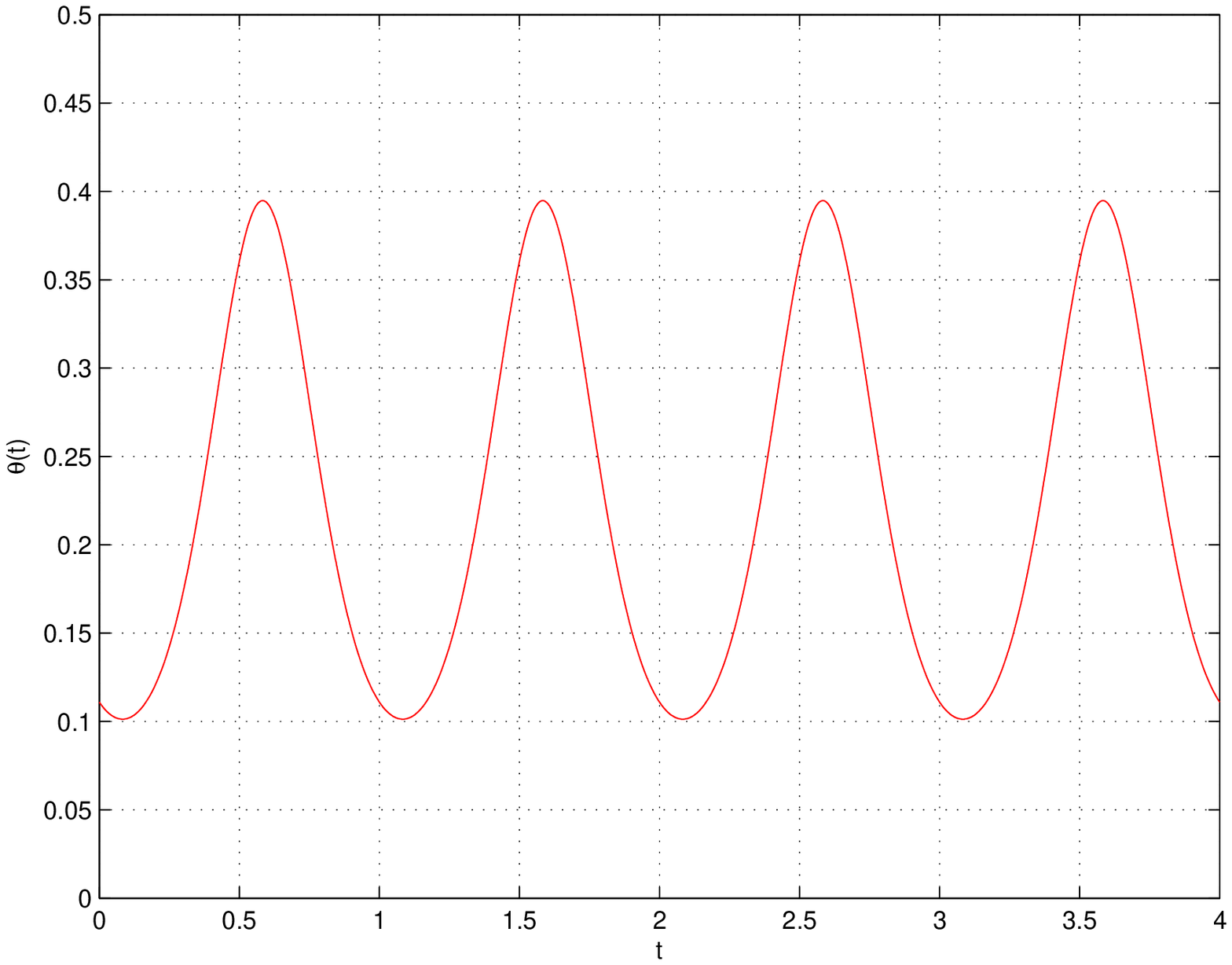}
	
	\includegraphics[height=5.0cm]{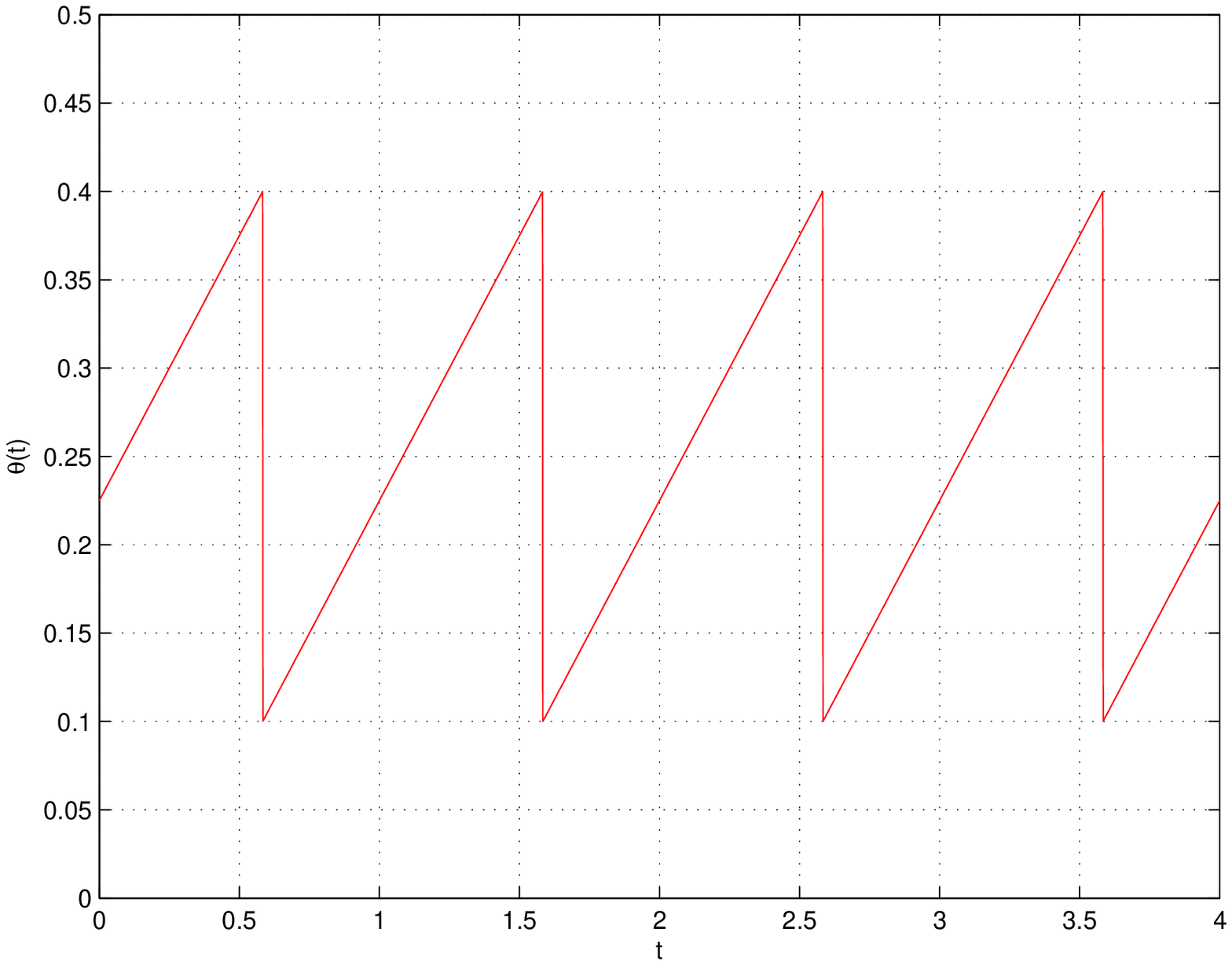} \quad
	\includegraphics[height=5.0cm]{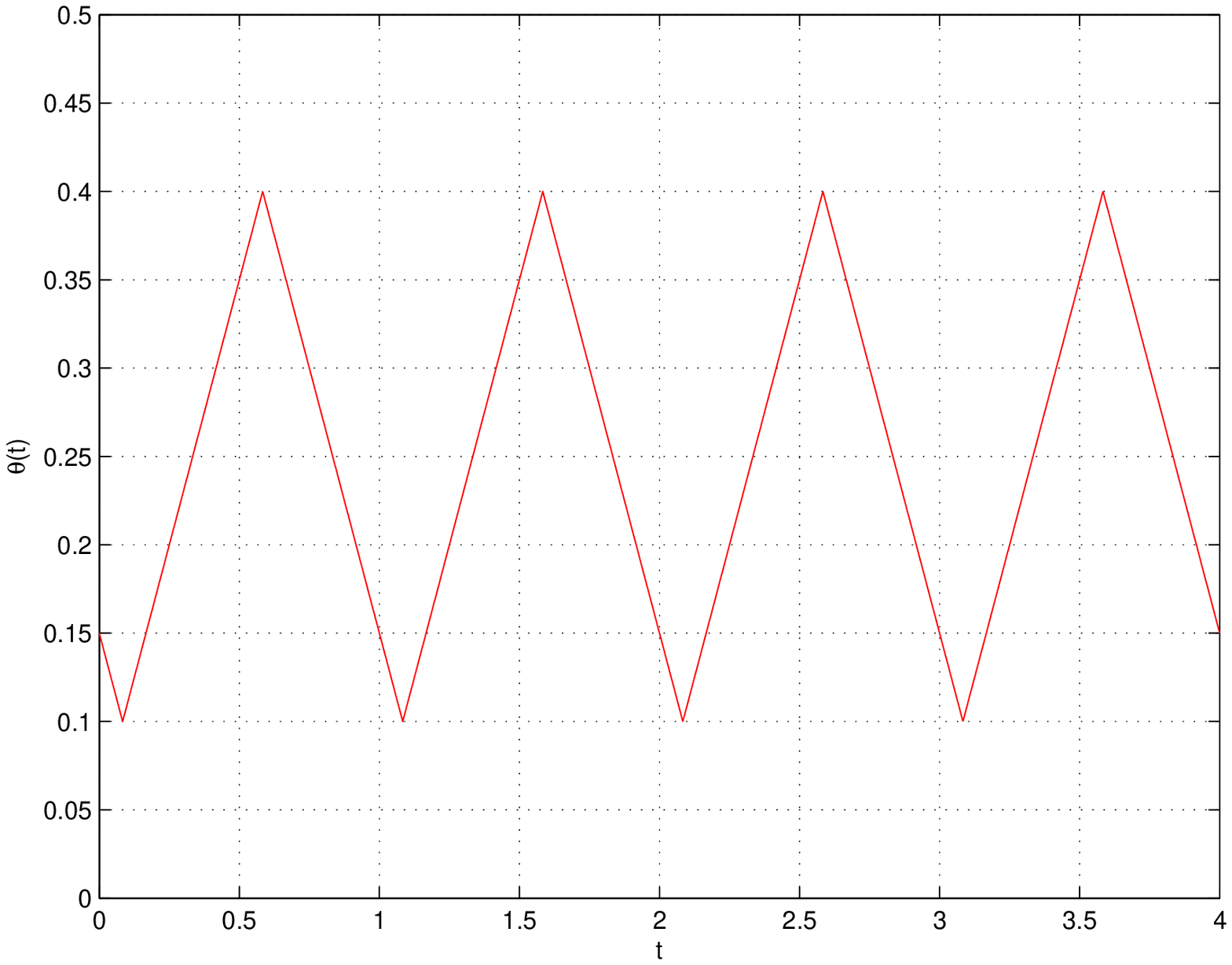}
	
	\includegraphics[height=5.0cm]{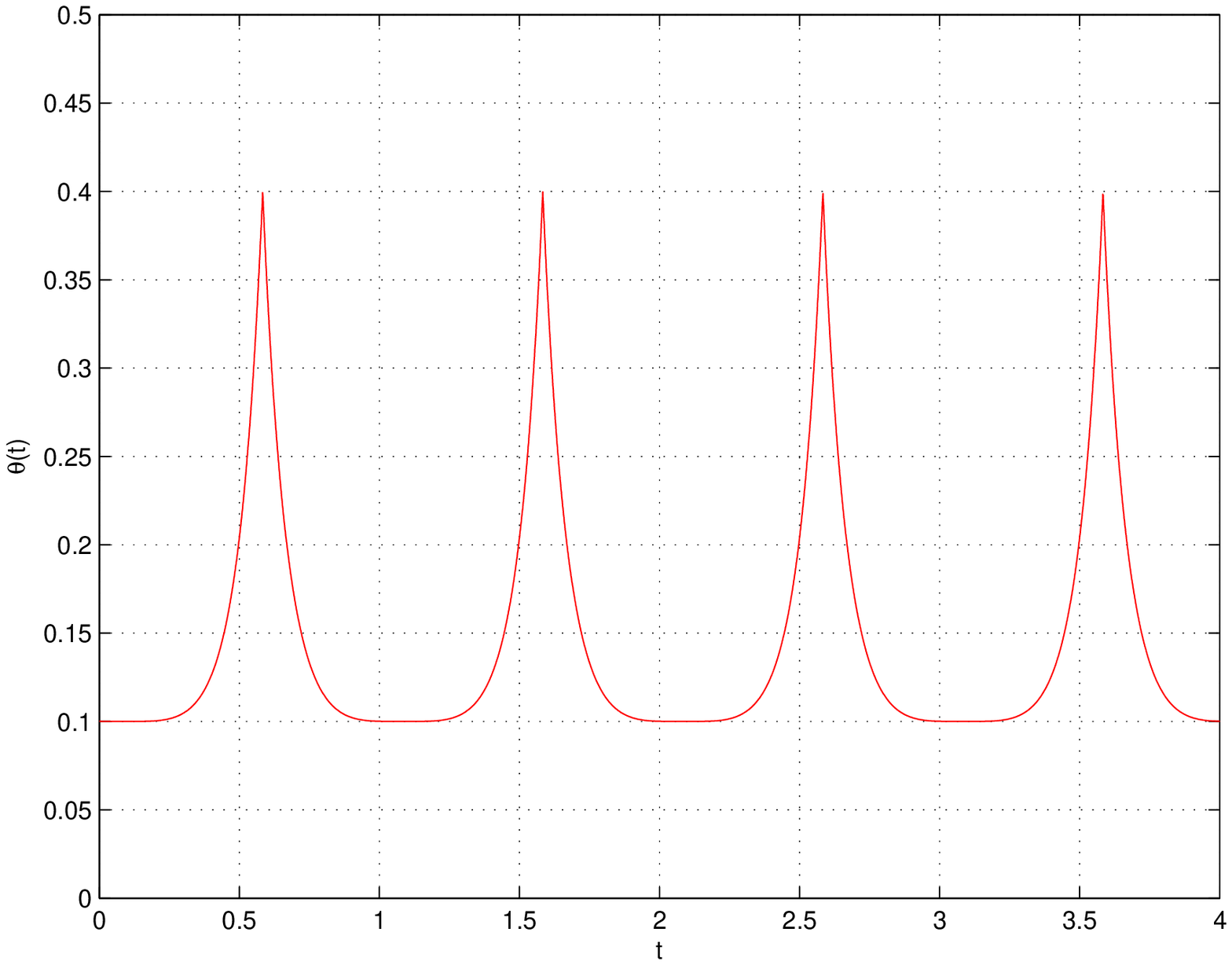}

	\caption{\label{Fig:theta_plots}
\textit{Upper left}: sinusoidal pattern. \textit{Upper right}: exponential-sinusoidal pattern.
\textit{Center left}: sawtooth pattern. \textit{Center right}: triangle pattern. \textit{Lower}: spiked pattern.}
\end{figure}


\section{Implied Volatility Smiles and Term-Structures}
\label{s:ImpliedVolatilities}

In this section we present implied volatility smiles and term-structures produced by a one-factor version of the proposed model with seasonality.
We have chosen the one-factor model here because our purpose is to illustrate the effect of seasonality on option prices.
The considered options are vanilla options on futures. Following the market convention, their maturity is the same as the maturity of the underlying futures contract.
The parameters of the considered model with seasonality are gathered in Table \ref{tab:ModelParameters}.
Only the parameters corresponding to the seasonality pattern change from one setting to another. The seasonality functions $\theta$ obtained with these parameters are those shown in Figure \ref{Fig:theta_plots}.
This numerical application is presented for an illustrative purpose and we take an initial futures curve that is flat in price at $100$ USD.

\begin{table}[htbp]
  \centering
  \caption{Model parameters for different specifications of the seasonality functions used in the numerical illustrations.}
    \begin{tabular}{cccccc}
    \addlinespace
    \toprule
		& \multicolumn{5}{c}{seasonality function} \\
		\cmidrule(lr){2-6}
		parameters & sinusoid & exp-sinusoid & sawtooth & triangle & spiked \\
		\midrule
		
		$v_0$ & $0.10$ & $0.10$ & $0.10$ & $0.10$ & $0.10$ \\
		$\lambda$ &$1.00$ & $1.00$ & $1.00$ & $1.00$ & $1.00$ \\
 		$\kappa$ &$0.80$ & $0.80$ & $0.80$ & $0.80$ & $0.80$ \\
		$\sigma$ &$1.20$ & $1.20$ & $1.20$ & $1.20$ & $1.20$ \\
		$\rho$ &$-0.25$ & $-0.25$ & $-0.25$ & $-0.25$ & $-0.25$ \\
		$a$ & $0.25$ & $0.20$ & $0.10$ & $0.10$ & $0.10$ \\
		$b$ & $0.15$ & $0.68$ & $0.30$ & $0.60$ & $0.30$ \\
		$t_0$ & $7/12$ & $7/12$ & $7/12$ & $7/12$ & $7/12$ \\
				
    \bottomrule
    \end{tabular}
  \label{tab:ModelParameters}
\end{table}

In Figures \ref{Fig:implied_vol1} and \ref{Fig:implied_vol2} we present, for the different seasonality patterns,
the implied volatility smiles obtained for different maturities and the term-structure of implied volatility for at-the-money options.
The obtained term-structures exhibit both seasonality and the Samuelson effect.
The obtained strike-structures exhibit smiled shapes.
The implied volatility produced with the sinusoidal and exponential-sinusoidal patterns are similar to each other.
For the sawtooth, triangle and spiked patterns, the irregularities of the function $\theta$ seem not to be transferred to the implied volatility.
It can also be observed that the choice of the seasonality pattern seems to have very little impact on the shape of the volatility smile.
This last remark is particularly striking if one compares the volatility term-structures produced by the sinusoidal and triangle patterns.

\begin{figure}[H]
\centering
	\includegraphics[height=5.0cm]{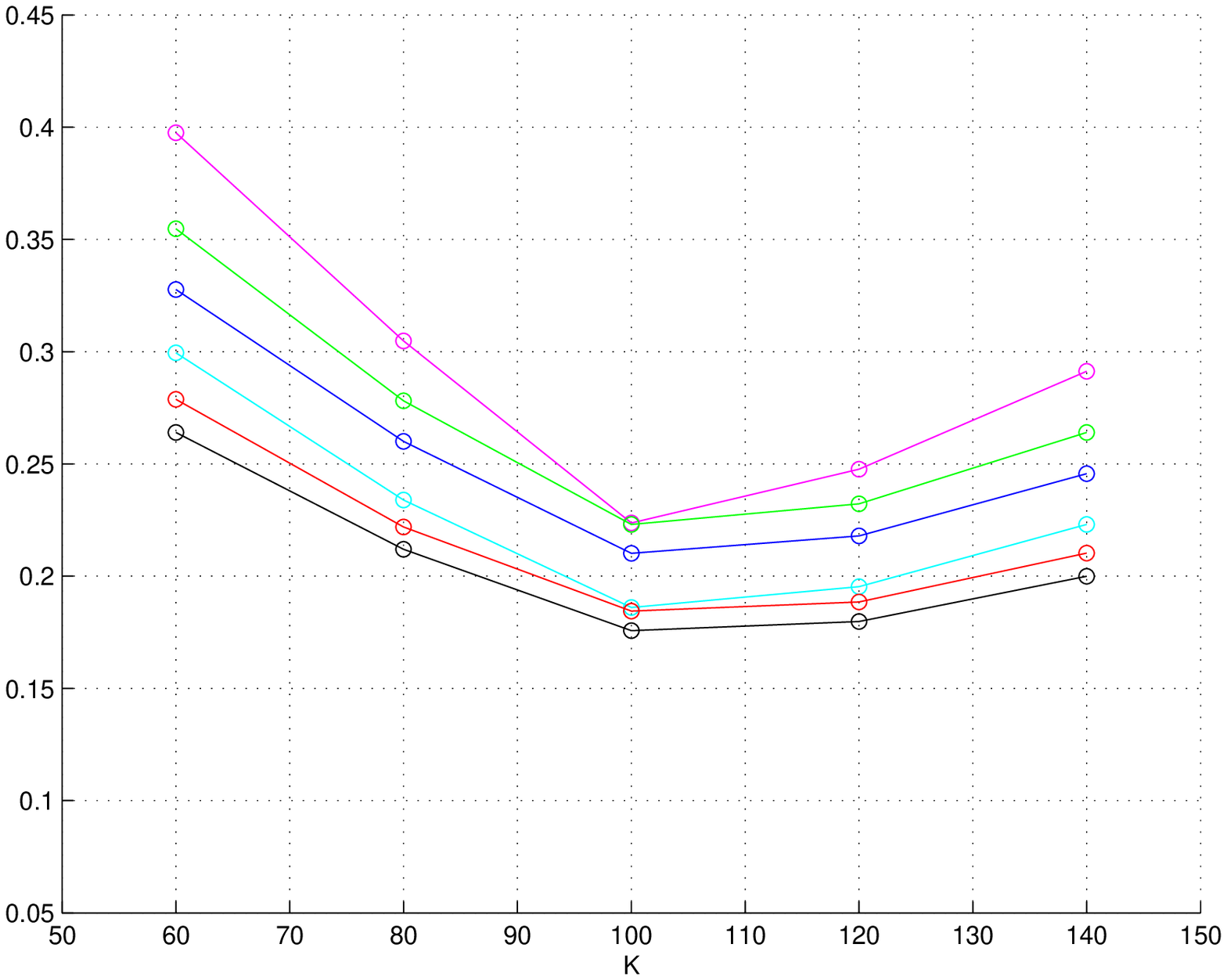} \quad
	\includegraphics[height=5.0cm]{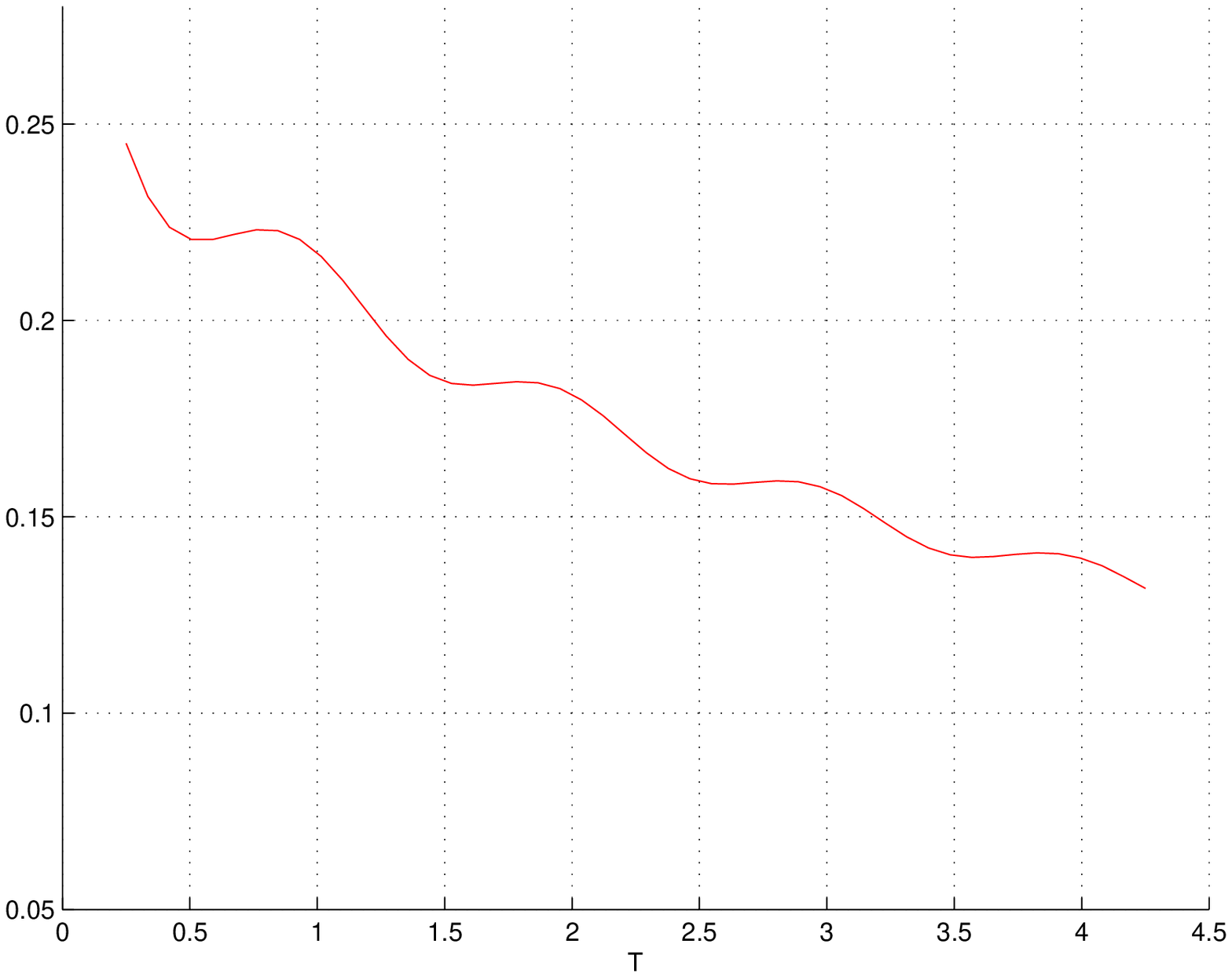}
	
	\includegraphics[height=5.0cm]{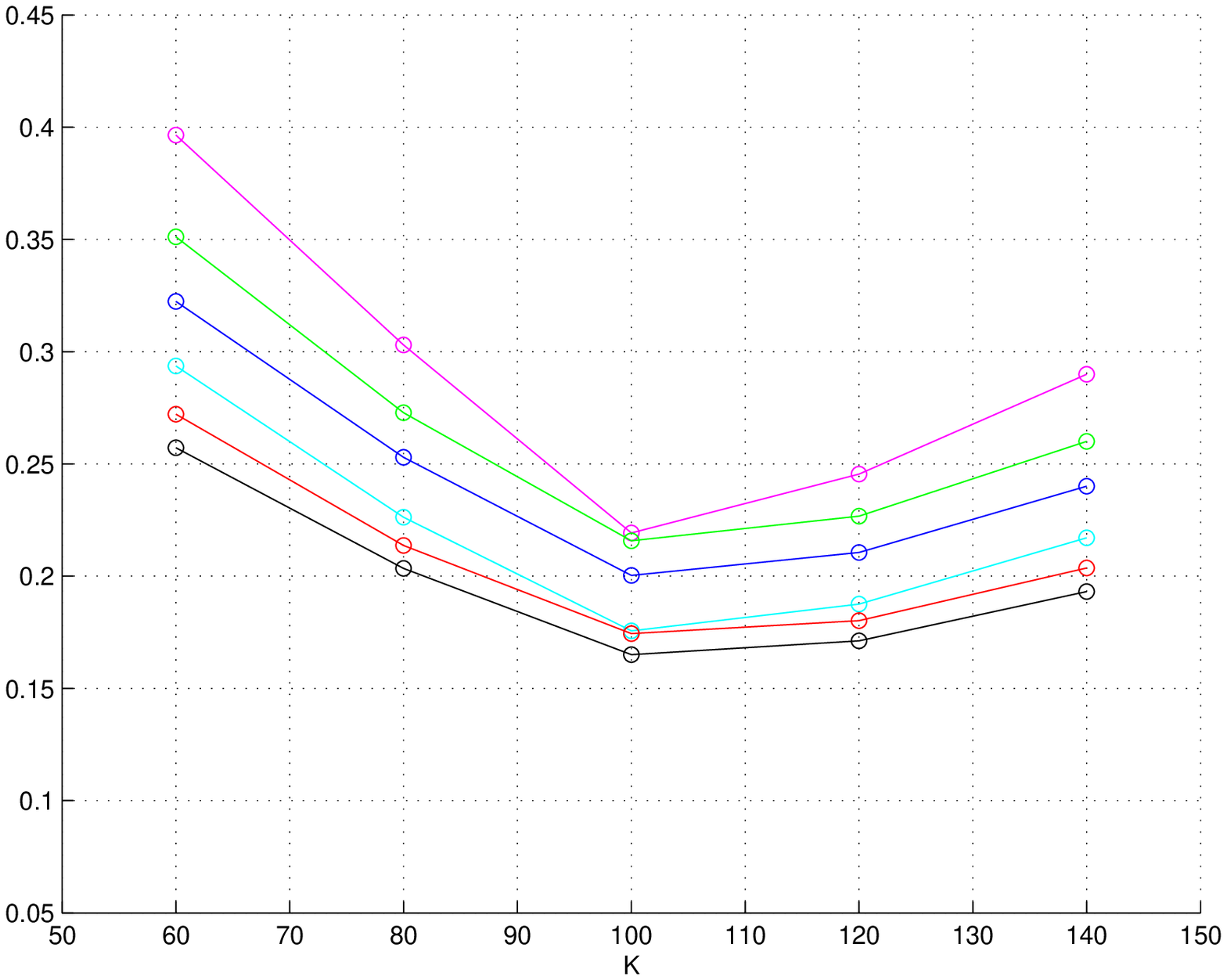} \quad
	\includegraphics[height=5.0cm]{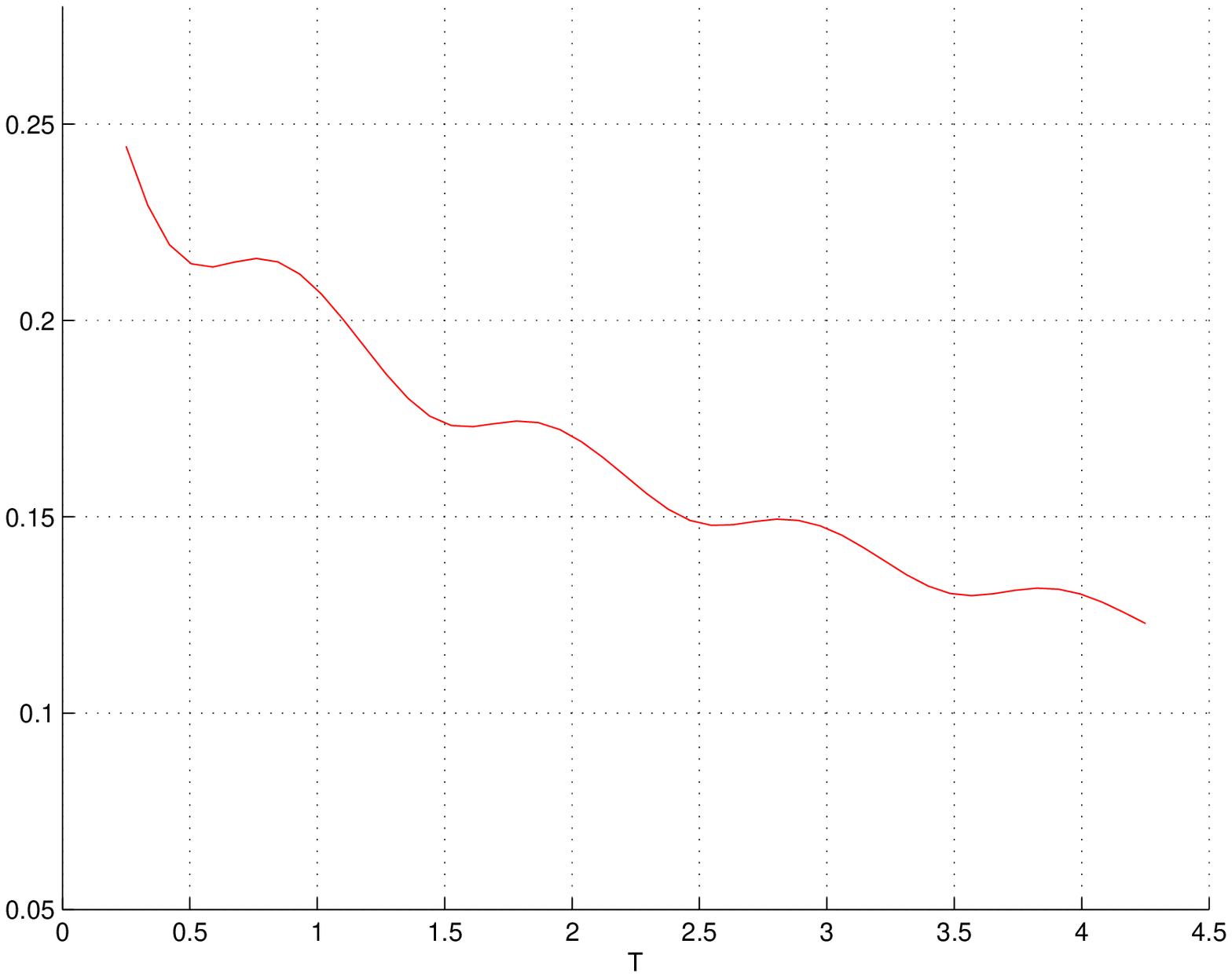}
	
		\caption{\label{Fig:implied_vol1}
\textit{Upper panel}: implied volatility smiles (left) and term-structure (right) obtained using the sinusoidal seasonality pattern with parameters in Table \ref{tab:ModelParameters}.
\textit{Lower panel}: implied volatility smiles (left) and term-structure (right) obtained using the exponential-sinusoidal pattern with parameters in Table \ref{tab:ModelParameters}.}
\end{figure}

\begin{figure}[H]
\centering
	\includegraphics[height=5.0cm]{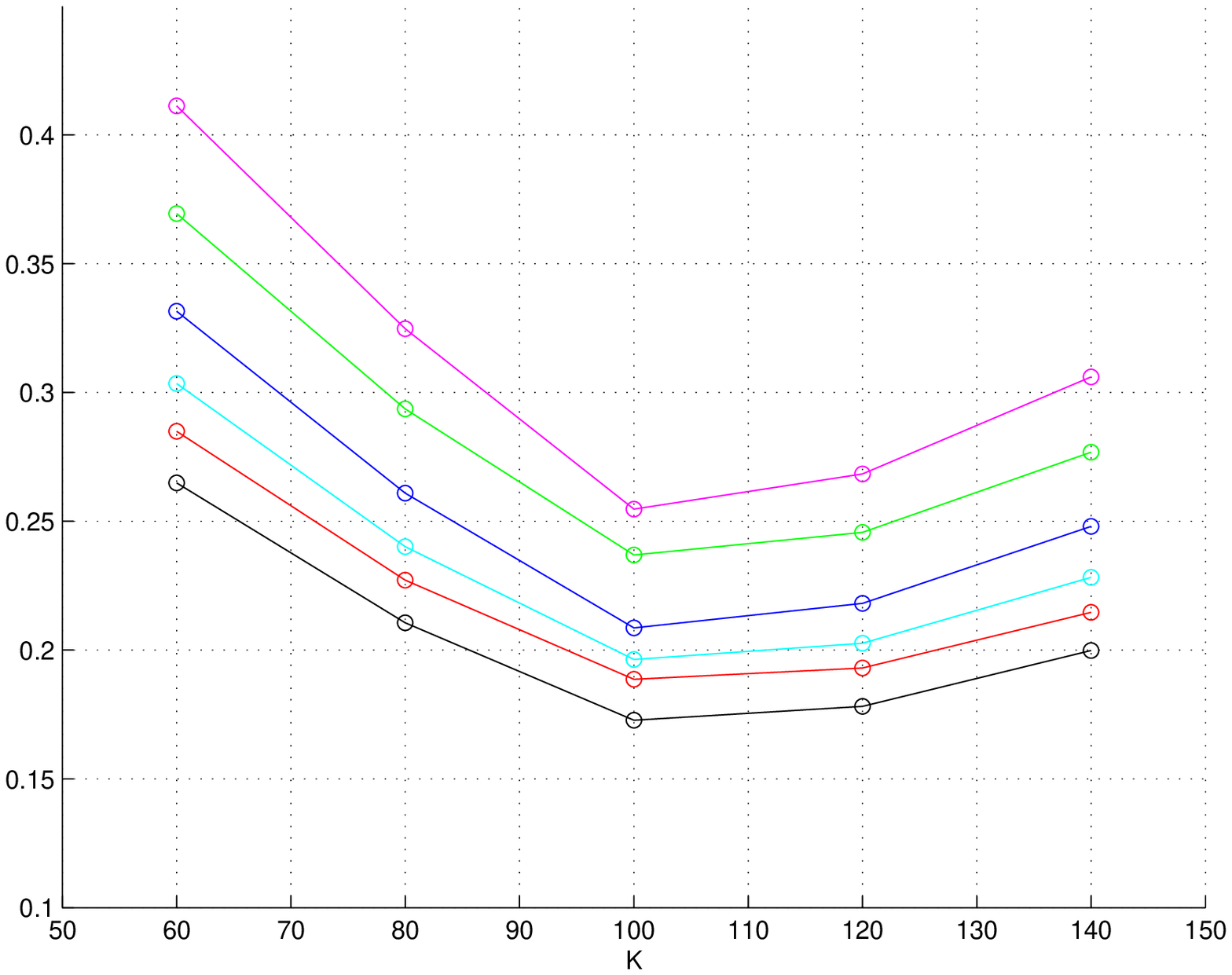} \quad
	\includegraphics[height=5.0cm]{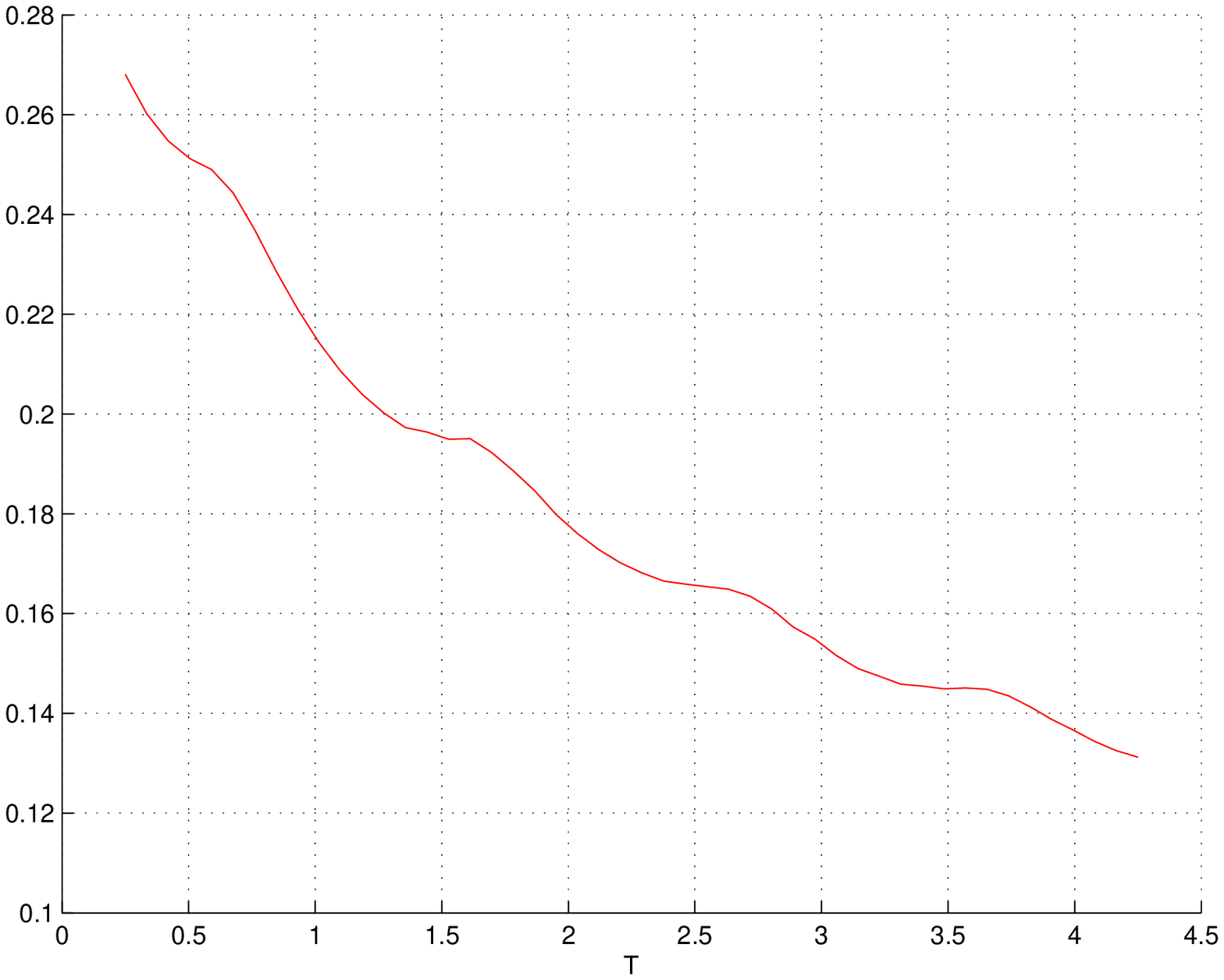}
	
	\includegraphics[height=5.0cm]{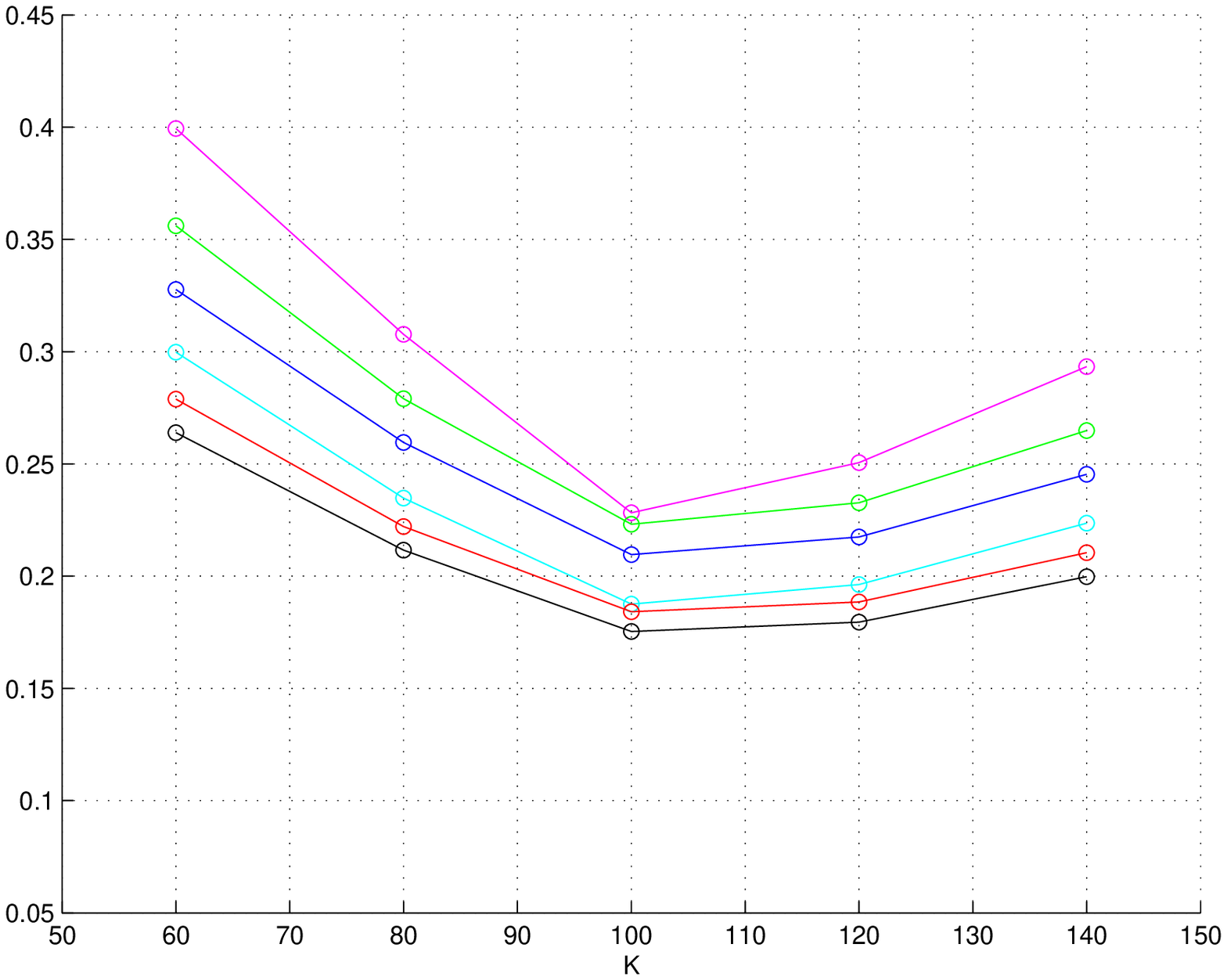} \quad
	\includegraphics[height=5.0cm]{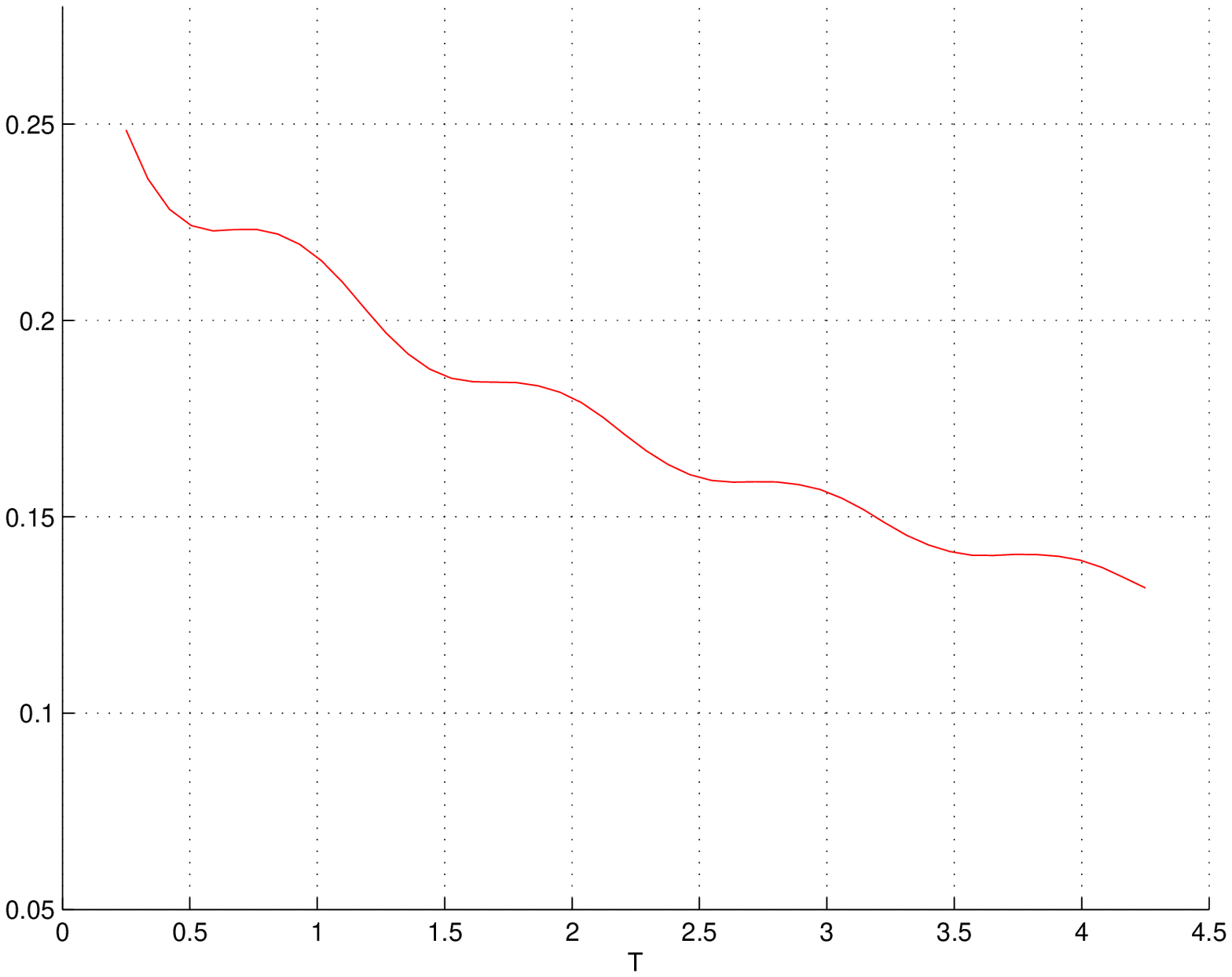}
	
	\includegraphics[height=5.0cm]{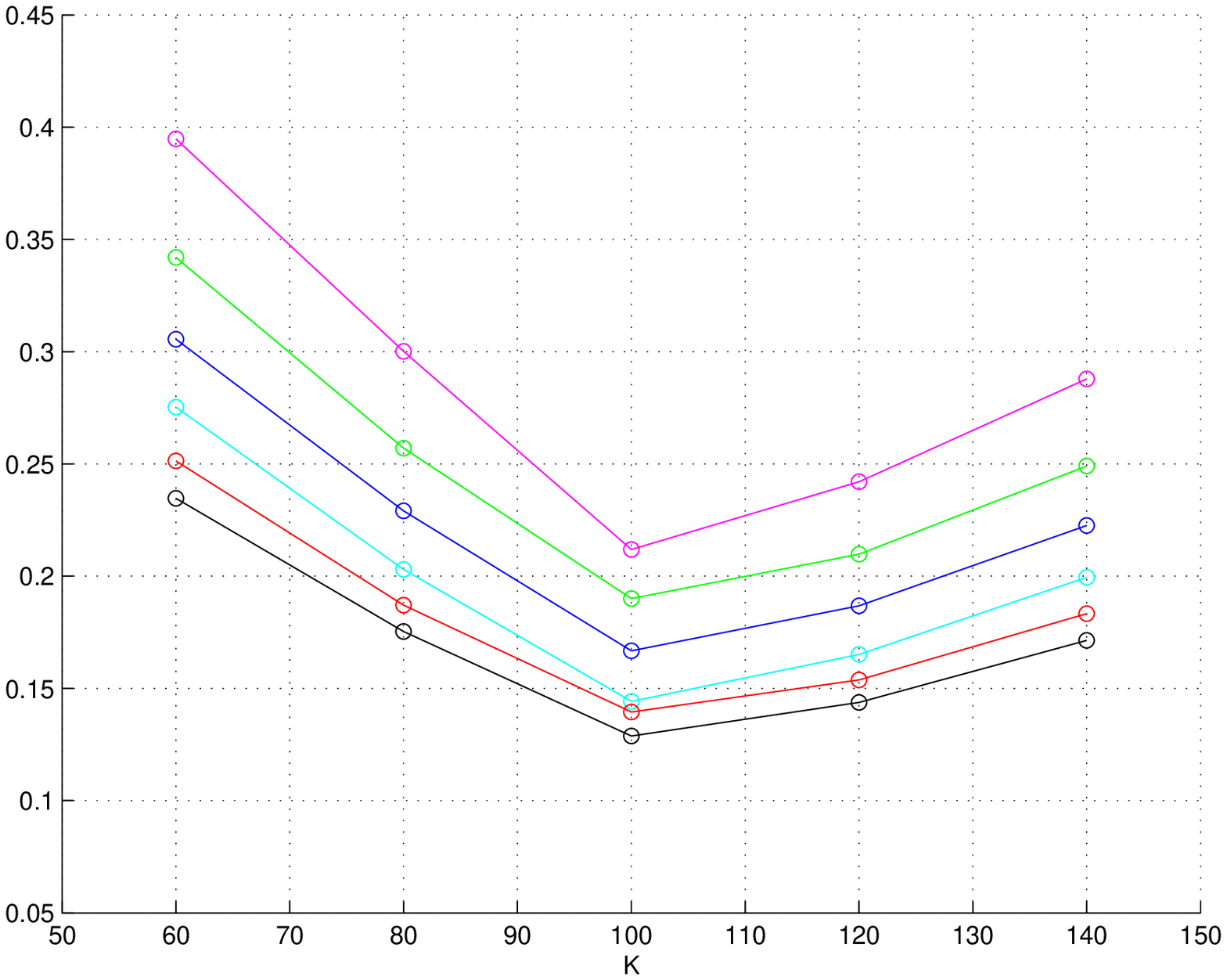} \quad
	\includegraphics[height=5.0cm]{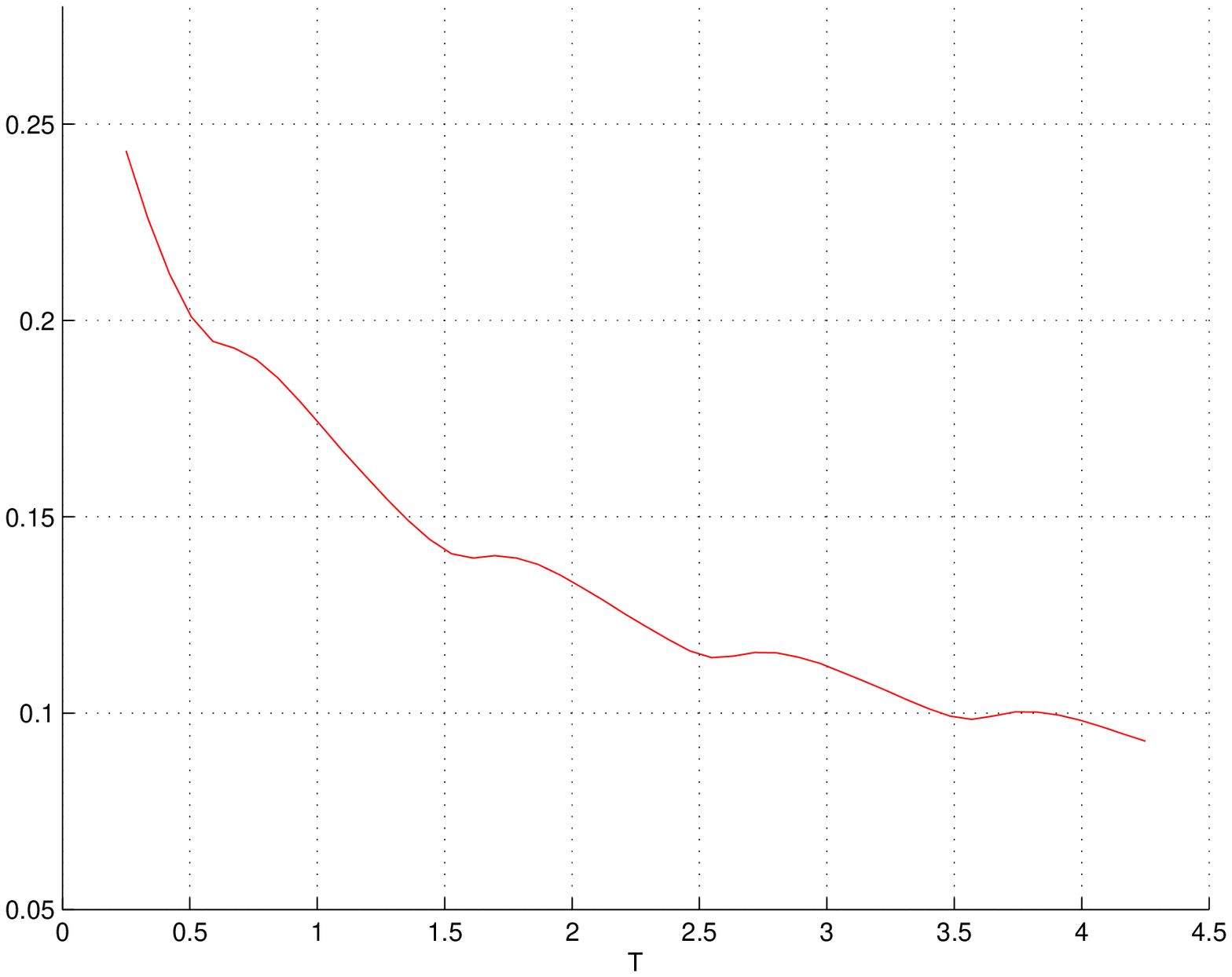}
	
	\caption{\label{Fig:implied_vol2}
\textit{Upper panel}: implied volatility smiles (left) and term-structure (right) obtained using the sawtooth seasonality pattern with parameters in Table \ref{tab:ModelParameters}.
\textit{Center panel}: implied volatility smiles (left) and term-structure (right) obtained using the triangle seasonality pattern with parameters in Table \ref{tab:ModelParameters}.
\textit{Lower panel}: implied volatility smiles (left) and term-structure (right) obtained using the spiked seasonality pattern with parameters in Table \ref{tab:ModelParameters}.}
\end{figure}

\section{Seasonal Stochastic Correlation and Calendar Spread Option Prices in the Multi-Factor Model}
\label{s:SeasonalStochasticCorrelation}

\subsection{Seasonal Stochastic Correlation}
\label{ss:SeasonalStochasticCorrelation}

We will show in this section that if we specify our model with two or more volatility factors,
then the correlation of the returns of two given futures contracts is influenced by the seasonality functions.
In other words, the correlation becomes seasonal.

Let us take the $2$-factor model for an illustration.
Futures returns follow the SDE
\begin{equation}
\label{FuturesSDE_2FactorModel}
\frac{dF(t, T_m)}{F(t, T_m)} = e^{-\lambda_1(T_m - t)} \sqrt{v_1(t)} dB_1(t) + e^{-\lambda_2(T_m - t)} \sqrt{v_2(t)} dB_2(t).
\end{equation}
and the two variance processes follow the SDEs
\begin{align}
\label{VarianceSDE1_2FactorModel}
dv_1(t) &= \kappa_1 \left( \theta_1(t) - v_1(t) \right) dt + \sigma_1 \sqrt{v_1(t)} dB_3(t),
\\
\label{VarianceSDE2_2FactorModel}
dv_2(t) &= \kappa_2 \left( \theta_2(t) - v_2(t) \right) dt + \sigma_2 \sqrt{v_2(t)} dB_4(t).
\end{align}
The correlations are given by $\langle dB_1(t), dB_3(t) \rangle = \rho_1 dt$, $\langle dB_2(t), dB_4(t) \rangle = \rho_2 dt$,
and all other correlations are zero.

Define
\begin{equation}
\label{Vij}
V_{ij}(t) := \langle \frac{dF(t, T_i)}{F(t, T_i)}, \frac{dF(t, T_j)}{F(t, T_j)} \rangle / dt.
\end{equation}
Then the instantaneous correlation $\rho(t)$ at time $t$ is given by:
\begin{equation}
\label{rho_tT1T2}
\rho(t) = \frac{V_{12}(t)}{\sqrt{V_{11}(t)} \sqrt{V_{22}(t)}}.
\end{equation}

Using \eqref{Vij} together with \eqref{FuturesSDE_2FactorModel} gives for the instantaneous correlation
\begin{equation}
\label{rho_2FactorModel}
\rho(t) =
\frac{e^{-\lambda_1(T_1 + T_2 - 2t)} v_1(t) + e^{-\lambda_2(T_1 + T_2 - 2t)} v_2(t)}
{\sqrt{e^{-2 \lambda_1(T_1 - t)} v_1(t) + e^{-2 \lambda_2(T_1 - t)} v_2(t)}
\sqrt{e^{-2 \lambda_1(T_2 - t)} v_1(t) + e^{-2 \lambda_2(T_2 - t)} v_2(t)}}.
\end{equation}
In contrast to the $1$-factor model, the instantaneous correlation $\rho(t)$ in the $n$-factor model, with $n \geq 2$, is stochastic.

To illustrate the seasonality of the correlation function $\rho$, we consider the case $\sigma_1 = \sigma_2 = 0$
when both variances $v_1$ and $v_2$ are deterministic functions of time $t$.
The two SDEs \eqref{VarianceSDE} then become ordinary differential equations
\begin{equation}
\label{VarianceODE}
v'_j(t) = \kappa_j \left( \theta_j(t) - v_j(t) \right), \; v_j(0) = v_{j,0} > 0.
\end{equation}
It is straightforward to see that the solution to \eqref{VarianceODE} is given by
\begin{align}
v_j(t) & = e^{-\kappa_j t} \left( v_{j,0} + \kappa_j \int_0^t e^{\kappa_j s} \theta_j(s) ds \right)
\\
\label{VarianceODESolution}
 & = e^{-\kappa_j t} \left( v_{j,0} + \kappa_j  \hat{\theta}_t(\kappa_j) \right).
\end{align}
Note that the same transform function $\hat{\theta}_t(\kappa) = \int_0^t e^{\kappa s} \theta(s) ds$
already introduced in \eqref{thetaTransform} appears again.
For several specifications of $\theta$, the function $\hat{\theta}_t$ is available in closed form, and it is easy
to plot the correlation function $\rho$ given by \eqref{rho_2FactorModel}.

Figure \ref{Fig:rho_inst1} presents the instantaneous correlation obtained with the two-factor model with seasonality in the special case $\sigma_1 = \sigma_2 = 0$.
It also plots the correlation obtained with the version without seasonality and the other terms involved in expression \eqref{rho_2FactorModel}.
For the seasonal model, only the first factor is seasonal and follows the sinusoidal pattern \eqref{sinusoid}.
To produce Figure \ref{Fig:rho_inst1}, we consider two cases in order to illustrate different seasonal patterns of the correlation when compared to the non-seasonal case.
The non-seasonal case is our benchmark, which is obtained by setting $b=0$ in the seasonality function of each factor.
The model parameters for these cases are gathered in Table \ref{tab:ModelParameters2}.

\begin{table}[htbp]
  \centering
  \caption{Model parameters in the two cases used to illustrate the seasonal behavior of the instantaneous correlation in the two-factor version of the model.}
    \begin{tabular}{ccc}
    \addlinespace
    \toprule
		parameters & Case 1 & Case 2 \\
		\midrule
		
		$v_{01}$ & $0.10$ & $0.06$ \\
		$v_{02}$ & $0.04$ & $0.04$ \\
		$\lambda_1$ &$2.00$ & $0.50$ \\
		$\lambda_2$ &$0.50$ & $2.00$ \\
 		$\kappa_1$ &$1.00$ & $1.00$ \\
		$\kappa_2$ &$1.00$ & $1.00$ \\
		$\sigma_1$ &$0.00$ & $0.00$ \\
		$\sigma_2$ &$0.00$ & $0.00$ \\
		$\rho_1$ &$0.00$ & $0.00$ \\
		$\rho_2$ &$0.00$ & $0.00$ \\
		$a_1$ & $0.10$ & $0.06$ \\
		$a_2$ & $0.04$ & $0.04$ \\
		$b_1$ & $0.09$ & $0.05$ \\
		$b_2$ & $0.00$ & $0.00$ \\
		$t_{01}$ & $0.00$ & $0.00$ \\
		$t_{02}$ & $0.00$ & $0.00$ \\
				
    \bottomrule
    \end{tabular}
  \label{tab:ModelParameters2}
\end{table}

\begin{figure}[H]
\centering
	\includegraphics[height=5.0cm]{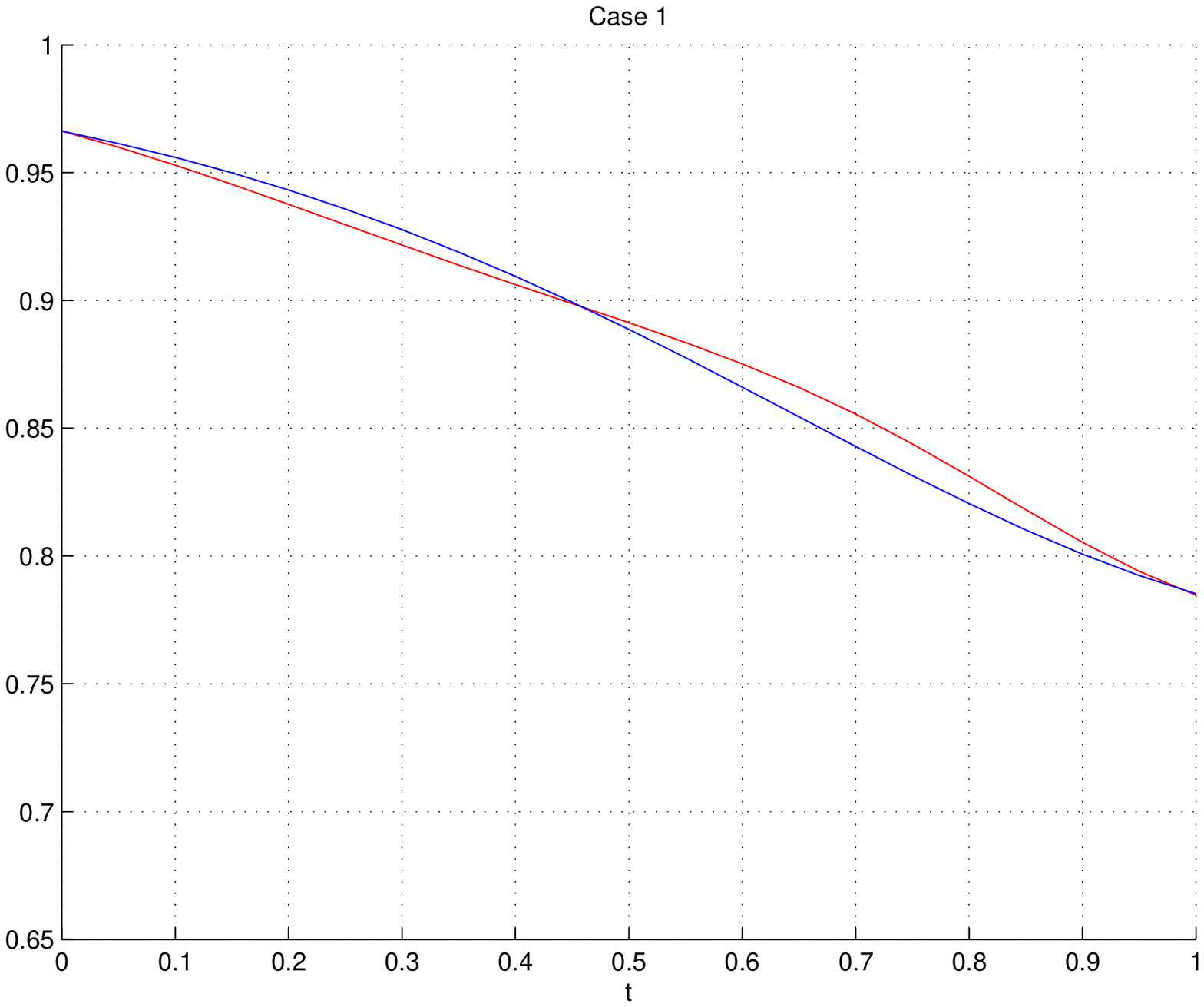} \quad
	\includegraphics[height=5.0cm]{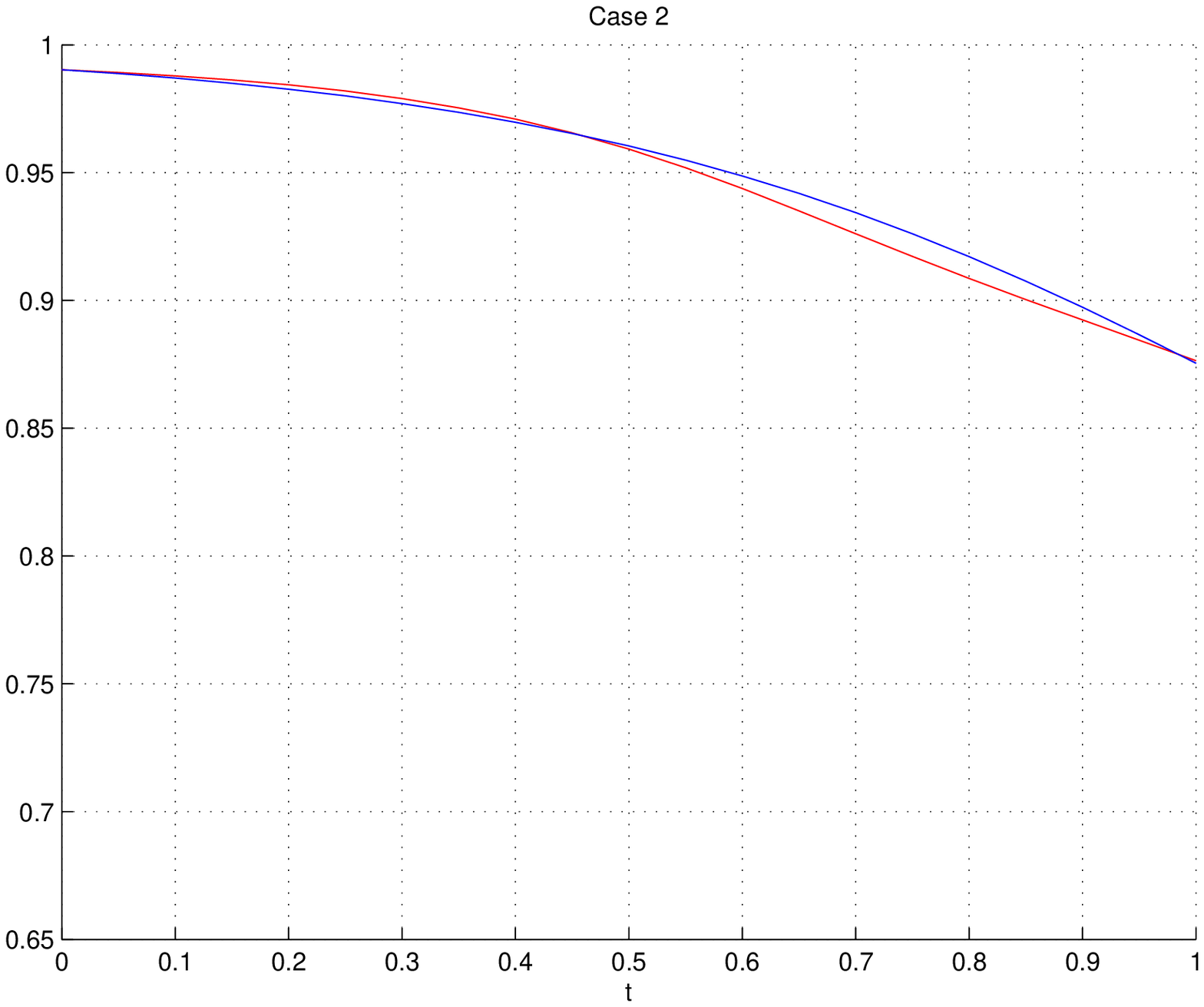}
	
	\includegraphics[height=5.0cm]{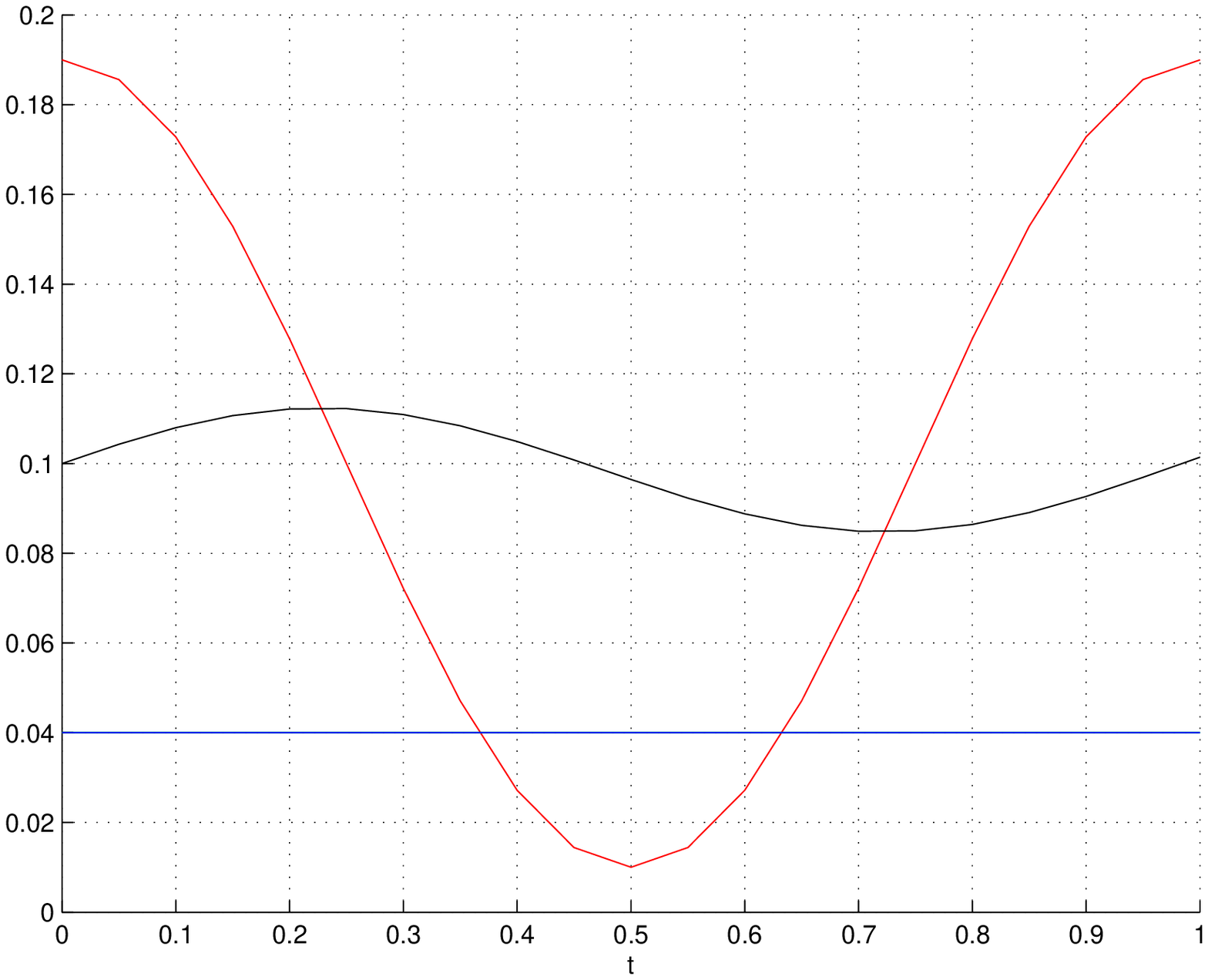} \quad
	\includegraphics[height=5.0cm]{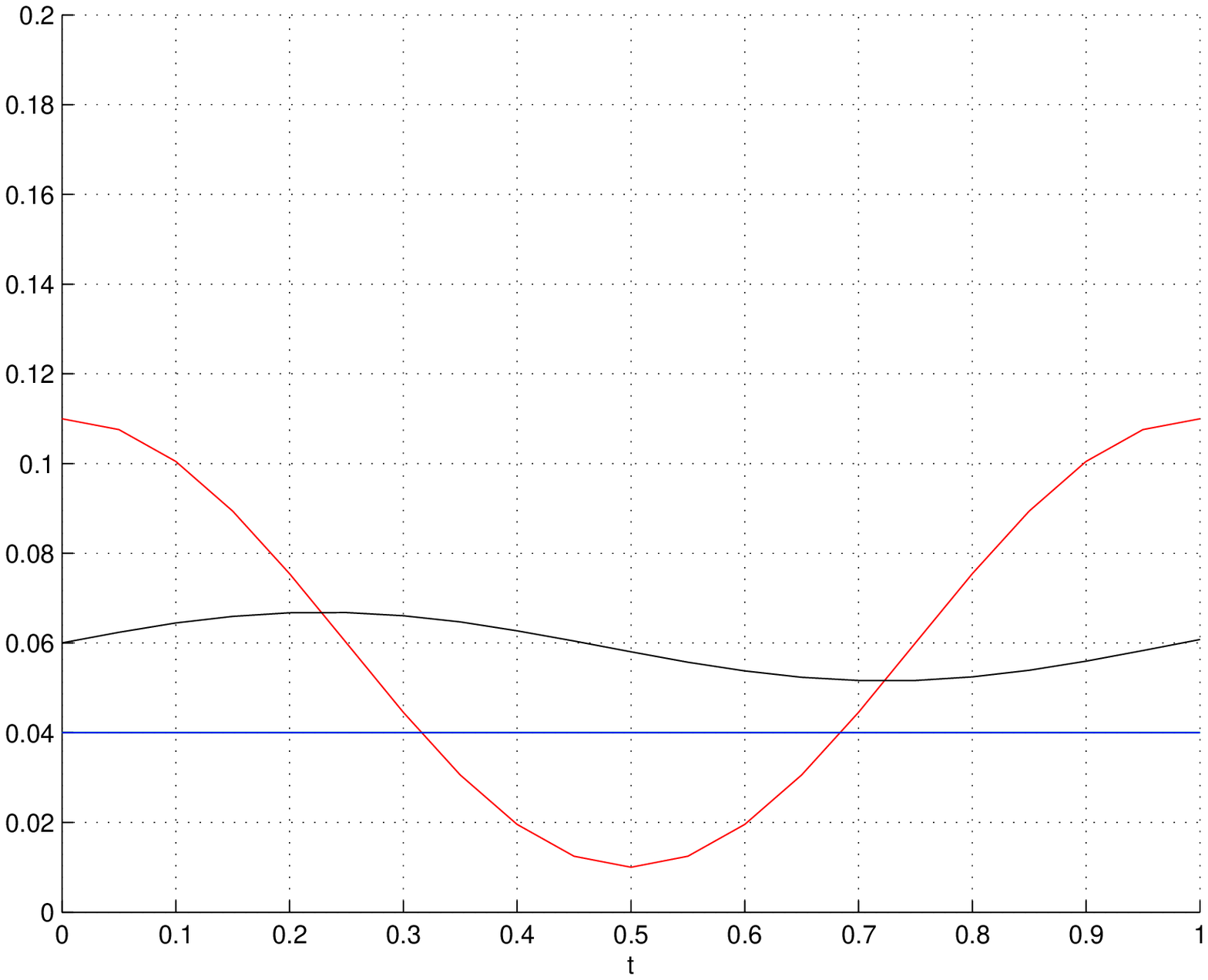}
	
	\includegraphics[height=5.0cm]{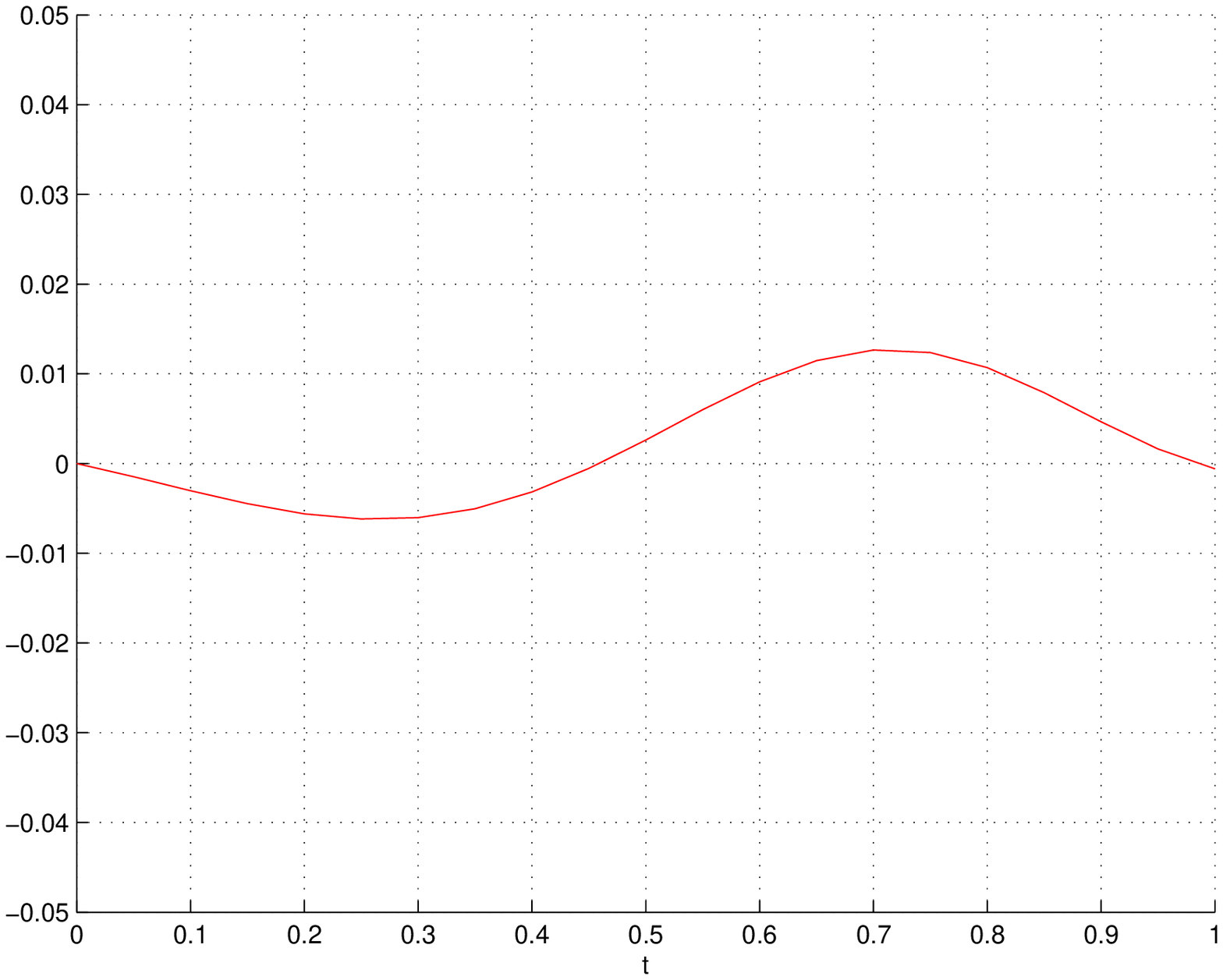} \quad
	\includegraphics[height=5.0cm]{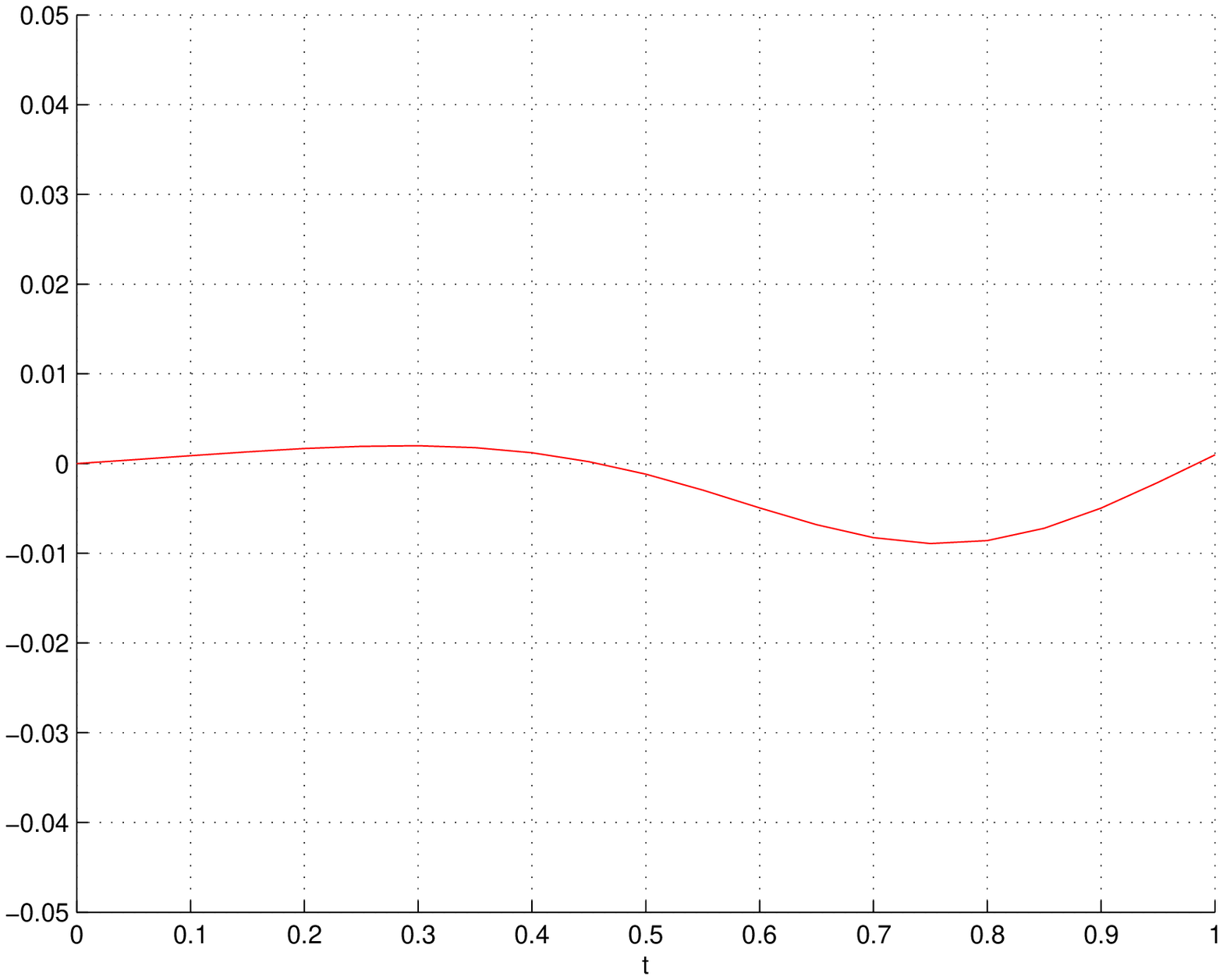}
	
	\caption{\label{Fig:rho_inst1}
\textit{Upper panel}: instantaneous correlation obtained with the two-factor model with seasonality (\textit{red}) and without seasonality (\textit{blue}) for the two cases detailed in Table \ref{tab:ModelParameters2} (case 1, \textit{left} and case 2, \textit{right}).
\textit{Center panel}: terms involved in expression \eqref{rho_2FactorModel}, $\theta_1(t)$ in \textit{red}, $v_1(t)$ in \textit{black} and $\theta_2(t)$  in \textit{green}, $v_2(t)$ in \textit{blue}, for the two cases detailed in Table \ref{tab:ModelParameters2} (case 1, \textit{left} and case 2, \textit{right}). $\theta_2$ and $v_2$ are identical as the second factor does not have seasonality.
\textit{Lower panel}: difference between instantaneous correlations obtained with and without seasonality for the two cases detailed in Table \ref{tab:ModelParameters2} (case 1, \textit{left} and case 2, \textit{right}).}
\end{figure}

It is worth remarking that, compared to the non-seasonal case, the instantaneous correlation is modified by the presence of seasonality in the dynamics of the first factor.
In the first case (high $\lambda$ for the seasonal factor) this correlation is lower at the beginning of the period under scrutiny and higher at the end.
In the second case (low $\lambda$ for the seasonal factor) it is the opposite.
The reasons underpinning the influence of $\lambda_1$ and $\lambda_2$ on the seasonal behaviour of the correlation are still conjectural and need to be further investigated.
So far, it seems on the one hand that if we have $\lambda_2 < 1 < \lambda_1$, then the effect of the first seasonal function on the instantaneous correlation is in the opposite direction,
i.e. during periods of high variance the correlation is decreased, and during periods of low variance the correlation is increased.
This pattern can be observed in the three l.h.s. panels of Figure \ref{Fig:rho_inst1} illustrating Case 1.
On the other hand, if we have $\lambda_1 < 1 < \lambda_2$, then the effect of the first seasonal function on the instantaneous correlation is in the same direction,
i.e. during periods of high variance the correlation is increased, and during periods of low variance the correlation is decreased.
This pattern can be observed in the three r.h.s. panels of Figure \ref{Fig:rho_inst1} illustrating Case 2.

Note also that the instantaneous correlation is not always decreasing with $t$.
This remark holds true for the model with seasonality as well as for the case without seasonality.
Figure \ref{Fig:rho_inst2} presents a non-decreasing instantaneous correlation obtained with the two-factor model with and without seasonality in the special case $\sigma_1 = \sigma_2 = 0$.
The other parameters are chosen such that the curves are non-decreasing.
In the presented illustration the correlation seems to be constant at first sight, but on closer inspection can be seen to be decreasing then increasing.

\begin{figure}[H]
\centering
	\includegraphics[height=5.0cm]{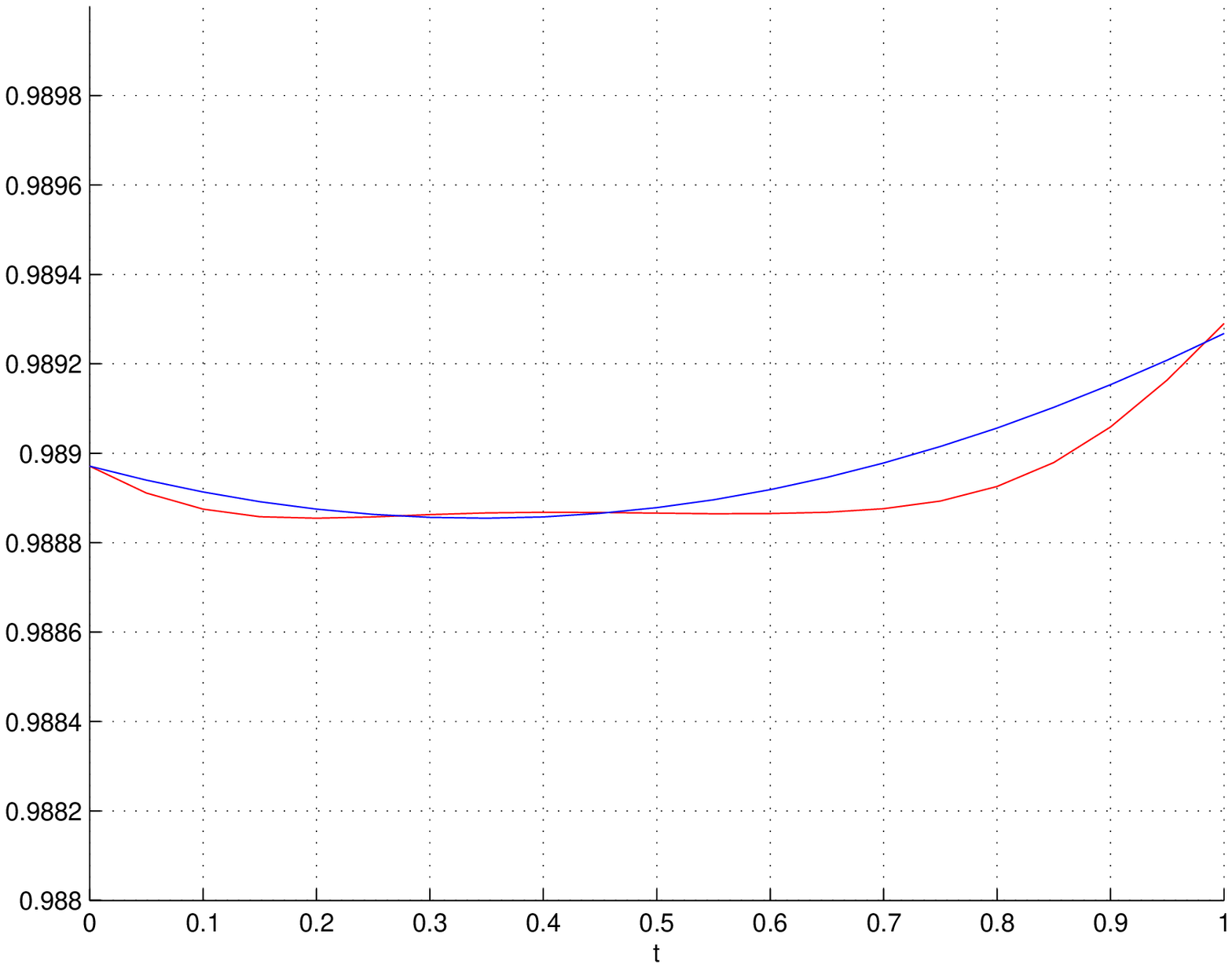}
	
	\caption{\label{Fig:rho_inst2}
Instantaneous correlation obtained with the two-factor model with seasonality (\textit{red}) and without seasonality (\textit{blue}).
Parameters are chosen such that both curves are non-decreasing.}
\end{figure}

\subsection{The Effect of Seasonality on Calendar Spread Option Prices}
\label{ss:EffectSeasonalityCalendarSpreadOptionPrices}

In this section, we investigate the effect of seasonality on the prices of calendar spread options.
We do so by means of a numerical study in which we consider different cases of model parameters and compute calendar spread option prices with and without seasonality.

This investigation is conducted with the two-factor model with stochastic volatility on both factors but seasonality only on the first one in order to isolate the effect of seasonality.
The form of the seasonal pattern is again the sinusoidal one given by equation \eqref{sinusoid}.
Calendar spread option prices are obtained with the \citet{CaldanaFusai2013} method that is well suited for our model.
Implied correlations are extracted using a numerical root search.

Table \ref{tab:SpreadOptionPrices1} presents, for different magnitudes of seasonality, calendar spread option prices obtained with the two-factor model.
Case 1 corresponds to the model without seasonality, Case 2 corresponds to a moderate seasonality and Case 3 to a stronger seasonality.
The model parameters for these three cases are shown in Table \ref{tab:ModelParameters3}.
Figure \ref{Fig:rho_imp1} plots the implied correlation term-structures obtained from calendar spread options in the considered cases.

Prices presented in Table \ref{tab:SpreadOptionPrices1} show the effect of seasonality on calendar spread options.
Irrespective of the strike, prices are increasing with the seasonality magnitude for the two maturities before and at a multiple of $t_0$ and decreasing for the two maturities after a multiple of $t_0$.

Figure \ref{Fig:rho_imp1} shows that seasonality has a noticeable effect on the implied correlation term-structure for calendar spread options.
Compared to the non-seasonal case, implied correlations are higher for maturities before a multiple of $t_0$ and lower for maturities after a multiple of $t_0$.

\begin{table}[htbp]
  \centering
	\small
  \caption{Calendar spread option prices obtained with the two-factor model with stochastic volatility and different magnitudes of seasonality on the first factor.
  The reported prices are for different strikes and maturities.
  The difference between the underlying futures maturities is held constant at $6$ months.
  The corresponding model parameters are shown in Table \ref{tab:ModelParameters3}.}
    \begin{tabular}{cccccccccccc}
    \addlinespace
    \toprule
		 & & & \multicolumn{3}{c}{Case 1: no seasonality} & \multicolumn{3}{c}{Case 2: moderate seasonality} & \multicolumn{3}{c}{Case 3: strong seasonality} \\
		\cmidrule(lr){4-6} \cmidrule(lr){7-9} \cmidrule(lr){10-12}
		$T$ & $T1$ & $T2$ & $K=-10$ & $K=0$ & $K=10$ & $K=-10$ & $K=0$ & $K=10$ & $K=-10$ & $K=0$ & $K=10$ \\
		\midrule
$0.33$ & $0.33$ & $0.83$ & $10.4381$ & $2.5993$ & $0.4951$ & $10.5605$ & $3.0931$ & $0.6822$ & $10.6647$ & $3.4190$ & $0.8351$ \\
$0.58$ & $0.58$ & $1.08$ & $10.6718$ & $3.4672$ & $0.9188$ & $10.8369$ & $3.7314$ & $1.1026$ & $10.9583$ & $3.9252$ & $1.2400$ \\
$0.83$ & $0.83$ & $1.33$ & $11.1648$ & $4.5229$ & $1.6366$ & $11.1425$ & $4.3690$ & $1.5636$ & $11.1291$ & $4.2429$ & $1.5120$ \\
$1.08$ & $1.08$ & $1.58$ & $11.3550$ & $4.6534$ & $1.8404$ & $11.2930$ & $4.5533$ & $1.7709$ & $11.2459$ & $4.4773$ & $1.7187$ \\
$1.33$ & $1.33$ & $1.83$ & $11.1894$ & $4.1480$ & $1.6234$ & $11.2674$ & $4.4534$ & $1.7695$ & $11.3323$ & $4.6612$ & $1.8843$ \\
$1.58$ & $1.58$ & $2.08$ & $11.1872$ & $4.4875$ & $1.7783$ & $11.3086$ & $4.6726$ & $1.9165$ & $11.3983$ & $4.8091$ & $2.0197$ \\
$1.83$ & $1.83$ & $2.33$ & $11.5326$ & $5.2374$ & $2.3126$ & $11.4832$ & $5.0666$ & $2.2079$ & $11.4488$ & $4.9289$ & $2.1312$ \\
$2.08$ & $2.08$ & $2.58$ & $11.6266$ & $5.2422$ & $2.3840$ & $11.5474$ & $5.1193$ & $2.2922$ & $11.4873$ & $5.0262$ & $2.2233$ \\
$2.33$ & $2.33$ & $2.83$ & $11.3966$ & $4.6788$ & $2.0671$ & $11.4623$ & $4.9321$ & $2.1982$ & $11.5160$ & $5.1052$ & $2.2992$ \\
$2.58$ & $2.58$ & $3.08$ & $11.3529$ & $4.8929$ & $2.1509$ & $11.4589$ & $5.0521$ & $2.2717$ & $11.5367$ & $5.1692$ & $2.3617$ \\
$2.83$ & $2.83$ & $3.33$ & $11.6528$ & $5.5345$ & $2.6146$ & $11.5933$ & $5.3604$ & $2.4988$ & $11.5511$ & $5.2210$ & $2.4132$ \\
				
    \bottomrule
    \end{tabular}
  \label{tab:SpreadOptionPrices1}
\end{table}

\begin{table}[htbp]
  \centering
  \caption{Model parameters in the three cases used to compute calendar spread option prices and the corresponding implied correlations.}
    \begin{tabular}{cccc}
    \addlinespace
    \toprule
		parameters & Case 1 & Case 2 & Case 3 \\
		\midrule
		
		$v_{01}$ & $0.10$ & $0.10$ & $0.10$ \\
		$v_{02}$ & $0.04$ & $0.04$ & $0.04$ \\
		$\lambda_1$ &$2.00$ & $2.00$ & $2.00$ \\
		$\lambda_2$ &$0.50$ & $0.50$ & $0.50$ \\
 		$\kappa_1$ &$0.80$ & $0.80$ & $0.80$ \\
		$\kappa_2$ &$0.80$ & $0.80$ & $0.80$ \\
		$\sigma_1$ &$1.20$ & $1.20$ & $1.20$ \\
		$\sigma_2$ &$0.90$ & $0.90$ & $0.90$ \\
		$\rho_1$ &$-0.25$ & $-0.25$ & $-0.25$ \\
		$\rho_2$ &$-0.25$ & $-0.25$ & $-0.25$ \\
		$a_1$ & $0.25$ & $0.25$ & $0.25$ \\
		$a_2$ & $0.10$ & $0.10$ & $0.10$ \\
		$b_1$ & $0.00$ & $0.15$ & $0.35$ \\
		$b_2$ & $0.00$ & $0.00$ & $0.00$ \\
		$t_{01}$ & $7/12$ & $7/12$ & $7/12$ \\
		$t_{02}$ & $7/12$ & $7/12$ & $7/12$ \\
				
    \bottomrule
    \end{tabular}
  \label{tab:ModelParameters3}
\end{table}

\begin{figure}[H]
\centering
	\includegraphics[height=6.0cm]{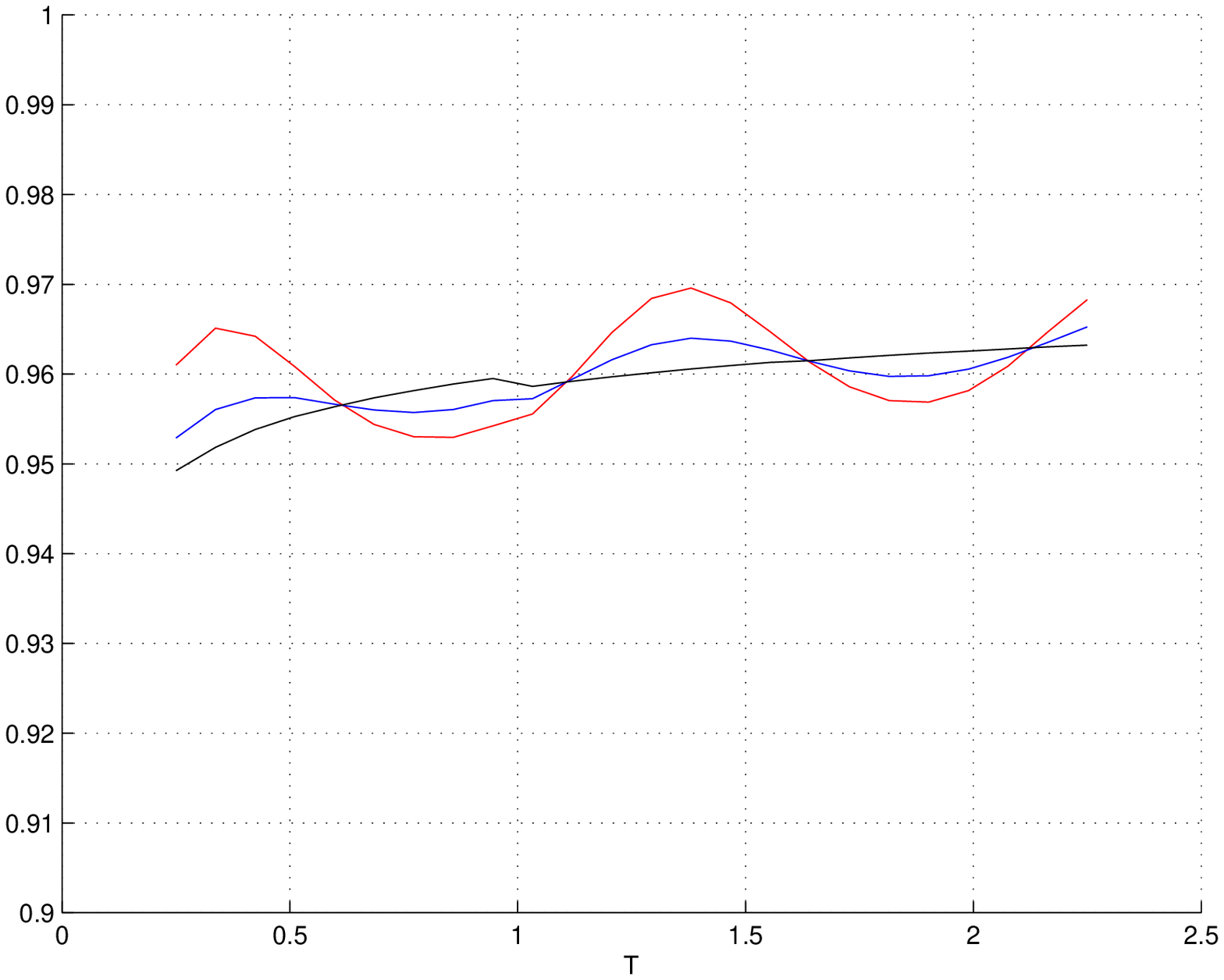}
	
	\caption{\label{Fig:rho_imp1}
Term-structure of implied correlation from calendar spread option prices obtained with the two-factor model with stochastic volatility and different magnitudes of seasonality on the first factor.
The different magnitudes considered are: no seasonality (\textit{black}), moderate seasonality (\textit{blue}) and strong seasonality (\textit{red}).
The corresponding model parameters are found in Table \ref{tab:ModelParameters3}.}
\end{figure}


\section{Conclusion}
\label{s:Conclusion}

We introduce a multi-factor seasonal stochastic volatility model for futures contracts that is capable of reproducing the Samuelson effect.
We show that the model can accommodate very general specifications of the seasonality functions, including only piece-wise continuous ones.
As an illustration, we suggest five different seasonality functions, some of which are familiar from the literature,
and provide details of how to incorporate these into the joint characteristic function of the model in a numerically fast and efficient way.

In a series of examples, we show that this model can reproduce seasonal implied volatility surfaces of European options on futures contracts.
Furthermore, implied correlations calculated from calendar spread option prices also show seasonal patterns.
Finally, we demonstrate that the instantaneous correlation between the returns of two futures contracts with different maturities
is both stochastic and seasonal, and make a conjecture about the relationship between the magnitudes of the Samuelson damping factors
and the effect of the seasonality functions on this correlation.


\pagebreak

\appendix


\section{Proofs}
\label{a1:Proofs}

In this appendix we prove Propositions \ref{Prop:CIR_SDE_Solution} and \ref{Prop:JointCharacteristicFunction}.
We also give the proofs of expressions (\ref{int_sinusoid}) and (\ref{int_sawtooth}).

\begin{MyProof}{of Proposition \ref{Prop:CIR_SDE_Solution}.}
\begin{enumerate}
\item
The drift function $b(t, v_t) := \kappa (\theta(t) - v(t))$ in \eqref{CIR_SDE_TimeDependentTheta} is Lipschitz continuous w.r.t the second argument, i.e.
$$
| b(t, x) - b(t, y) | \leq K | x - y |,
$$
where we can choose $K = \kappa$, since $| b(t, x) - b(t, y) | = \kappa | x - y |$.
Then Proposition 2.13 (Yamada and Watanabe) of \citet{KaratzasShreve1988} guarantees the existence of a unique strong solution to \eqref{CIR_SDE_TimeDependentTheta}
with continuous sample paths.
\item
The comparison result given in Proposition 2.18 of \citet{KaratzasShreve1988} establishes $v_t \geq \tilde{v}_t$ a.s. for all $t \geq 0$
under the hypothesis that the drift function $b(t, v_t)$ is continuous.
Now, if $\theta$ has a discontinuity at time $t_1$, we know from this argument applied to the interval $[0, t_1[$ that $\tilde{v}_{t} \leq v_{t} \forall t \in [0, t_1[$ (a.s.).
It then follows from the continuity of the sample paths that $\tilde{v}_{t_1} \leq v_{t_1}$ (a.s.),
and we can apply the argument again to the interval $]t_1, t_2[$ to obtain $\tilde{v}_{t} \leq v_{t} \forall t \in ]t_1, t_2[$ (a.s.).
Since by assumption the set $\mathcal{T}$ of times where $\theta$ has discontinuities has no limit points,
we can proceed in this manner to cover all of $\mathbb{R}_0^+$.
\item
The Feller condition $\sigma^2 < 2 \kappa \theta_{min}$ for $\theta_{min}$ implies the strict positivity a.s. of $\tilde{v}$.
The strict positivity of $v$ itself therefore follows immediately from (ii).
\end{enumerate}
\end{MyProof}

\begin{MyProof}{of Proposition \ref{Prop:JointCharacteristicFunction}.}
The proof is an extension of the proof of Proposition 2.1 of \citet{SchneiderTavin2015} to the case where the variance mean-reversion level $\theta$ is time-dependent.
Going from $\theta$ to $\theta(t)$ leads to changes in two places.
The first is in Lemma A.1 of \citet{SchneiderTavin2015}, which needs to be modified as follows.
\begin{lemma}
\label{Lemma:StochasticIntegral}
Let $\theta: \mathbb{R}_0^+ \to \mathbb{R}^+$ be the seasonal mean-reversion level function,
and let
$$
\hat{\theta}_T(\lambda) := \int_0^T e^{\lambda t} \theta(t) dt
$$
be its transform. Then
\begin{equation}
\label{StochasticIntegral}
\sigma \int_0^T f_1(t) \sqrt{v(t)} d\tilde{B}(t)
=
\left[ f_1(t) v(t) \right]_0^T - f_1(0) \kappa \hat{\theta}_T(\lambda) + (\kappa - \lambda) \int_0^T f_1(t) v(t) dt.
\end{equation}
\end{lemma}
\begin{proof}
Multiplying equation \eqref{VarianceSDE} by $f_1(t)$ and then integrating from $0$ to $T$ gives
\begin{equation}
\label{dv-integral1}
\int_0^T f_1(t) dv(t) = \int_0^T f_1(t) \kappa (\theta(t) - v(t)) dt + \sigma \int_0^T f_1(t) \sqrt{v(t)} d\tilde{B}(t).
\end{equation}
Using It\^{o}-integration by parts (see \cite{Oksendal2003}), we also have
\begin{align}
\int_0^T f_1(t) dv(t)
&= \left[ f_1(t) v(t) \right]_0^T - \int_0^T v(t) df_1(t)
\nonumber
\\
&= \left[ f_1(t) v(t) \right]_0^T - \lambda \int_0^T f_1(t) v(t) dt.
\label{dv-integral2}
\end{align}
Equating the right hand sides of equations \eqref{dv-integral1} and \eqref{dv-integral2} gives
\begin{align*}
\sigma \int_0^T f_1(t) \sqrt{v(t)} d\tilde{B}(t)
&= \left[ f_1(t) v(t) \right]_0^T - \lambda \int_0^T f_1(t) v(t) dt - \int_0^T f_1(t) \kappa (\theta(t) - v(t)) dt
\\
&= \left[ f_1(t) v(t) \right]_0^T - \kappa \int_0^T f_1(t) \theta(t) dt + (\kappa - \lambda) \int_0^T f_1(t) v(t) dt
\\
&= \left[ f_1(t) v(t) \right]_0^T - f_1(0) \kappa \int_0^T e^{\lambda t} \theta(t) dt + (\kappa - \lambda) \int_0^T f_1(t) v(t) dt
\\
&= \left[ f_1(t) v(t) \right]_0^T - f_1(0) \kappa \hat{\theta}_T(\lambda) + (\kappa - \lambda) \int_0^T f_1(t) v(t) dt,
\end{align*}
which proves the lemma.
\end{proof}

The second change in the proof is due to the appearance of $\theta$ in the generator of the process $v$.
As in \citet{SchneiderTavin2015}, let the function $h$ be given by
\begin{equation*}
h(t,v) = {\mathbb E} \left[ \exp \left( i \frac{\rho}{\sigma} f_1(T) v(T) + \int_t^T q(s) v(s) ds \right) \right].
\end{equation*}
Now $h$ satisfies the PDE
\begin{equation}
\label{hPDE}
\frac{\partial h}{\partial t} (t,v)
+ \kappa (\theta(t) - v(t)) \frac{\partial h}{\partial v} (t,v)
+ \frac{1}{2} \sigma^2 v(t) \frac{\partial^2 h}{\partial v^2} (t,v)
+ q(t) v(t) h(t,v)
= 0,
\end{equation}
with terminal condition
\begin{equation*}
 h(T,v) = \exp \left( i \frac{\rho}{\sigma} f_1(T) v(T) \right).
\end{equation*}

Again, we know from \citet{DuffiePanSingleton2000} that $h$ has affine form
\begin{equation}
\label{hGuess}
h(t,v) = \exp \left( A(t,T) v(t) + B(t,T) \right),
\end{equation}
with $A(T,T) = i \frac{\rho}{\sigma} f_1(T), B(T,T) = 0.$
Putting \eqref{hGuess} in \eqref{hPDE} gives
\begin{equation*}
B_t + A_t v + \kappa (\theta(t) - v) A + \frac{1}{2} \sigma^2 v A^2 + q v = 0,
\end{equation*}
and collecting the terms with and without $v$ leads to the two ODEs
\begin{align}
\label{A_RiccatiEquation}
A_t - \kappa A + \frac{1}{2} \sigma^2 A^2 + q &= 0,
\\
\label{A_Primitive}
B_t + \kappa \theta(t) A &= 0.
\end{align}
This completes the proof of the proposition.
\end{MyProof}

Note that $\theta$ only appears in the second ODE \eqref{A_Primitive},
and that therefore the closed-form expression previously given for $A$ in \citet{SchneiderTavin2015} can still be used.
Only the function $B$ changes due to the time-dependence of $\theta$.

\begin{MyProof}{of expression (\ref{int_sinusoid}). Transform of the sinusoidal pattern.}

With $y=t-t_0$
\begin{align}
\hat{\theta}_{T}(\lambda) & = \int^{T}_{0}{\left(a + b\cos{\left(2\pi\left(t - t_0\right)\right)} \right)e^{\lambda t}dt} \\
 & = \frac{a}{\lambda}\left(e^{\lambda T} - 1\right) + b e^{\lambda t_0}\int^{T-t_0}_{-t_0}{\cos{\left(2\pi y \right)} e^{\lambda y}dy}.
\end{align}
A primitive of $y \mapsto \cos{\left(2\pi y \right)} e^{\lambda y}$ is
\begin{equation}
y \mapsto \frac{e^{\lambda y}}{\lambda^2 + 4 \pi^2}\left(\lambda\cos{\left(2\pi y \right)} + 2 \pi \sin{\left(2\pi y \right)} \right),
\end{equation}
and
\begin{align}
\int^{T-t_0}_{-t_0}{\cos{\left(2\pi y \right)} e^{\lambda y}dy} & = \left[\frac{e^{\lambda y}}{\lambda^2 + 4 \pi^2}\left(\lambda\cos{\left(2\pi y \right)} + 2 \pi \sin{\left(2\pi y \right)} \right) \right]^{T-t_0}_{-t_0} \\
 & = \frac{e^{\lambda (T-t_0)}}{\lambda^2 + 4 \pi^2}\left(2 \pi \sin{\left(2\pi (T-t_0) \right)} + \lambda\cos{\left(2\pi (T-t_0) \right)} \right) \\
 & + \frac{e^{-\lambda t_0}}{\lambda^2 + 4 \pi^2}\left(2 \pi \sin{\left(2\pi t_0 \right)} - \lambda\cos{\left(2\pi t_0 \right)} \right)
\end{align}

\end{MyProof}

\begin{MyProof}{of expression (\ref{int_sawtooth}). Transform of the sawtooth pattern.}

With $T\geq 0$ and $t_0 \in [0,1[$
\begin{align}
\hat{\theta}_{T}(\lambda) & = \int^{T}_{0}{\left(a + b\left(t-t_0-\left\lfloor t-t_0 \right\rfloor \right) \right)e^{\lambda t}dt} \\
 & = \int^{T}_{0}{\left(a + b\left(t-t_0 \right) \right)e^{\lambda t}dt} -b \int^{T}_{0}{\left\lfloor t-t_0 \right\rfloor e^{\lambda t}dt}.
\end{align}

The first integral is computed as
\begin{equation}
\int^{T}_{0}{\left(a + b\left(t-t_0 \right) \right)e^{\lambda t}dt} = \frac{1}{\lambda}\left(b\left(\frac{1}{\lambda}+t_0 \right)-a \right)+\frac{e^{\lambda T}}{\lambda}\left(a+b\left(T-\frac{1}{\lambda}-t_0 \right) \right).
\end{equation}

The integral involving the floor function can be split, with $y=t-t_0$, as
\begin{equation}
\int^{T}_{0}{\left\lfloor t-t_0 \right\rfloor e^{\lambda t}dt} = \int^{T-t_0}_{0}{\left\lfloor y \right\rfloor e^{\lambda \left(y+t_0 \right)}dy} + \int^{0}_{-t_0}{\left\lfloor y \right\rfloor e^{\lambda \left(y+t_0 \right)}dy}.
\end{equation}

Noting that $\left\lfloor y \right\rfloor = -1$ for $y \in [-t_0,0[$ we have
\begin{equation}
\int^{0}_{-t_0}{\left\lfloor y \right\rfloor e^{\lambda \left(y+t_0 \right)}dy} = \frac{1}{\lambda}\left(1-e^{\lambda t_0} \right).
\end{equation}

The other part of the term with the floor function can be written, when $T\geq t_0$, as
\begin{align}
\int^{T-t_0}_{0}{\left\lfloor y \right\rfloor e^{\lambda \left(y+t_0 \right)}dy} & = e^{\lambda t_0} \left(\sum^{n-1}_{k=0}{k\int^{k+1}_{k}{e^{\lambda y}dy}} + n \int^{n+\alpha}_{n}{e^{\lambda y}dy} \right), \nonumber \\
 & = \frac{e^{\lambda t_0}}{\lambda} \left( \sum^{n-1}_{k=0}{k\left(e^{\lambda (k+1)}- e^{\lambda k}\right)} + n \left(e^{\lambda (n+\alpha)}- e^{\lambda n}  \right) \right), \nonumber \\
 & = \frac{e^{\lambda t_0}}{\lambda} \left( n e^{\lambda (n+\alpha)} - \sum^{n}_{k=1}{e^{\lambda k} }\right),
\end{align}
 with $n = \left\lfloor T-t_0 \right\rfloor$ and $\alpha = T-t_0-\left\lfloor T-t_0 \right\rfloor$.

When $0 \leq T < t_0$, as $t_0<1$, $\left\lfloor T-t_0 \right\rfloor = -1$ and $\left\lfloor y \right\rfloor = -1$ for $y \in [T-t_0,0[$ so that we have
\begin{equation}
\int^{T-t_0}_{0}{\left\lfloor y \right\rfloor e^{\lambda \left(y+t_0 \right)}dy} = \frac{e^{\lambda t_0}}{\lambda}\left(1 - e^{\lambda \left( T- t_0 \right)} \right).
\end{equation}

Gathering the components, the integral involving the floor function can now be written as
\begin{equation}
\int^{T}_{0}{\left\lfloor t-t_0 \right\rfloor e^{\lambda t}dt} = \frac{e^{\lambda t_0}}{\lambda}\left(\left\lfloor T-t_0 \right\rfloor e^{\lambda(T-t_0)} - \left(\sum^{\left\lfloor T-t_0 \right\rfloor}_{k=1}{e^{\lambda k}}\right)\mathbb{I}_{\left\{ T \geq t_0\right\}} + \mathbb{I}_{\left\{ T<t_0\right\}}+e^{-\lambda t_0}-1 \right).
\end{equation}

\end{MyProof}

\begin{MyProof}{of expression (\ref{int_triangle}). Transform of the triangle pattern.}

With $T\geq 0$ and $t_0 \in [0,1[$
\begin{align}
\hat{\theta}_{T}(\lambda) & = \int^{T}_{0}{\left(a + b\left|\frac{1}{2}-\left(t-t_0-\left\lfloor t-t_0 \right\rfloor \right) \right| \right) e^{\lambda t}dt} \\
 & = \frac{a}{\lambda}\left(e^{\lambda T}-1 \right) + b\int^{T}_{0}{\left|\frac{1}{2}-\left(t-t_0-\left\lfloor t-t_0 \right\rfloor \right) \right| e^{\lambda t}dt}.
\end{align}

With $y=t-t_0$, the last integral becomes
\begin{equation}
\int^{T}_{0}{\left|\frac{1}{2}-\left(t-t_0-\left\lfloor t-t_0 \right\rfloor \right) \right| e^{\lambda t}dt} = e^{\lambda t_0} \left(\int^{0}_{-t_0}{\left|\frac{1}{2}-\left(y-\left\lfloor y \right\rfloor \right) \right| e^{\lambda y}dy} + \int^{T-t_0}_{0}{\left|\frac{1}{2}-\left(y-\left\lfloor y \right\rfloor \right) \right| e^{\lambda y}dy} \right).
\end{equation}

Two integrals remain to be computed, one on $[-t_0,0]$ and the other on $[0,T-t_0]$.
To compute the first, one needs to distinguish two cases.
When $t_0 \in [0,\frac{1}{2}]$, it can be computed as
\begin{equation}
\int^{0}_{-t_0}{\left|\frac{1}{2}-\left(y-\left\lfloor y \right\rfloor \right) \right| e^{\lambda y}dy} = \frac{1}{\lambda}\left(z_2 - \left(z_2-t_0 \right)e^{-\lambda t_0}\right),
\end{equation}
and when $t_0 \in ]\frac{1}{2}, 1[$, it is
\begin{equation}
\int^{0}_{-t_0}{\left|\frac{1}{2}-\left(y-\left\lfloor y \right\rfloor \right) \right| e^{\lambda y}dy} = \frac{1}{\lambda}\left(z_2 + \frac{2}{\lambda}e^{-\frac{\lambda}{2}} + \left(z_2-t_0 \right)e^{-\lambda t_0}\right),
\end{equation}
with $z_2 = \frac{1}{2}-\frac{1}{\lambda}$.
To compute the other integral, on $[0,T-t_0]$, one needs first to distinguish two cases. First, when $T \geq t_0$, we have
\begin{equation}
\int^{T-t_0}_{0}{\left|\frac{1}{2}-\left(y-\left\lfloor y \right\rfloor \right) \right| e^{\lambda y}dy} = \sum^{n-1}_{k=0}{\int^{k+1}_{k}{\left|\frac{1}{2}-\left(y-\left\lfloor y \right\rfloor \right) \right| e^{\lambda y}dy}} + \int^{n+\alpha}_{n}{\left|\frac{1}{2}-\left(y-\left\lfloor y \right\rfloor \right) \right| e^{\lambda y}dy},
\end{equation}
with $n = \left\lfloor T-t_0 \right\rfloor$ and $\alpha = T-t_0 - \left\lfloor T-t_0 \right\rfloor$. For $k=1,\dots, n-1$, the integral in the sum can be computed as
\begin{align}
\int^{k+1}_{k}{\left|\frac{1}{2}-\left(y-\left\lfloor y \right\rfloor \right) \right| e^{\lambda y}dy} & = \int^{k+\frac{1}{2}}_{k}{\left|\frac{1}{2}-\left(y-\left\lfloor y \right\rfloor \right) \right| e^{\lambda y}dy} + \int^{k+1}_{k+\frac{1}{2}}{\left|\frac{1}{2}-\left(y-\left\lfloor y \right\rfloor \right) \right| e^{\lambda y}dy} \\
 & = \int^{k+\frac{1}{2}}_{k}{\left(\frac{1}{2}-y+k \right) e^{\lambda y}dy} - \int^{k+1}_{k+\frac{1}{2}}{\left(\frac{1}{2}-y+k \right) e^{\lambda y}dy} \\
& = \frac{1}{\lambda}\left(\frac{2}{\lambda}e^{\lambda\left(k+\frac{1}{2}\right)} + z_2 e^{\lambda(k+1)} - z_1 e^{\lambda k}  \right),
\end{align}

where $z_2=\frac{1}{2}+\frac{1}{\lambda}$. The sum becomes
\begin{equation}
\sum^{n-1}_{k=0}{\int^{k+1}_{k}{\left|\frac{1}{2}-\left(y-\left\lfloor y \right\rfloor \right) \right| e^{\lambda y}dy}} =  \frac{1}{\lambda}\left(\frac{2}{\lambda}e^{\frac{\lambda}{2}} + z_2 e^{\lambda} - z_1  \right) \sum^{n-1}_{k=0}{e^{\lambda k}}.
\end{equation}

The integral on $[n,n+\alpha]$ is computed,  as
\begin{equation}
\int^{n+\alpha}_{n}{\left|\frac{1}{2}-\left(y-\left\lfloor y \right\rfloor \right) \right| e^{\lambda y}dy} = \frac{e^{\lambda n}}{\lambda}\left(z_3 e^{\lambda \alpha} \mathbb{I}_{\left\{\alpha \leq \frac{1}{2} \right\}} + \left(\frac{2}{\lambda}e^{\frac{\lambda}{2}} - z_3 e^{\lambda \alpha} \right) \mathbb{I}_{\left\{\alpha > \frac{1}{2} \right\}} - z_1\right),
\end{equation}
with $z_3 = z_1 - \alpha$.

The second case is when $T \in [0,t_0[$. In this case, the integral on $[0,T-t_0]$ becomes
\begin{align}
\int^{T-t_0}_{0}{\left|\frac{1}{2}-\left(y-\left\lfloor y \right\rfloor \right) \right| e^{\lambda y}dy} & = \frac{1}{\lambda}\left(\left(e^{\lambda(T-t_0)}\left(z_2+T-t_0 \right)-z_2 \right)\mathbb{I}_{\left\{T-t_0 \in [-\frac{1}{2},0[ \right\}} \right. \nonumber \\
 & \left. -\left(\frac{2}{\lambda}e^{-\frac{\lambda}{2}} + e^{\lambda(T-t_0)}\left(z_2+T-t_0 \right) + z_2 \right)\mathbb{I}_{\left\{T-t_0 \in [-1,-\frac{1}{2}[ \right\}} \right).
\end{align}
Gathering the components now gives the result.
\end{MyProof}


\bibliographystyle{plainnat}

\bibliography{articles,books,websites}

\end{document}